\documentclass{article}

\usepackage{arxiv}

\usepackage[utf8]{inputenc} 
\usepackage[T1]{fontenc}    
\usepackage{hyperref}       
\usepackage{url}            
\usepackage{booktabs}       
\usepackage{amsfonts}       
\usepackage{nicefrac}       
\usepackage{microtype}      
\usepackage{lipsum}

\usepackage{epsfig} 
\usepackage{mathptmx} 
\usepackage{amsmath} 
\usepackage{amssymb}  
\usepackage{subfigure}
\usepackage{enumerate}
\newtheorem{remark}{Remark} 
\newtheorem{theorem}{Theorem}
\numberwithin{theorem}{subsection} 
\newtheorem{lemma}{Lemma} 
\numberwithin{lemma}{subsection} 
\newtheorem{proof}{Proof}
\numberwithin{proof}{subsection} 
\usepackage{xcolor}

\title{Model based fractional order controller design for process plants satisfying desired robustness criteria}

\author{
  Pushkar Prakash Arya
\\
  Department of Electrical Engineering\\
  Indian Institute of Technology Roorkee\\
  Roorkee, INDIA \\
  \texttt{parya@ee.iitr.ac.in} \\
   \And
Sohom Chakrabarty \\
  Department of Electrical Engineering\\
  Indian Institute of Technology Roorkee\\
Roorkee, INDIA \\
\texttt{sohomfee@iitr.ac.in} \\
}

\begin{document}
\maketitle

\begin{abstract}
This paper contributes to the design of a fractional order (FO) internal model controller (IMC) for a first order plus time delay (FOPTD) process model to satisfy a given set of desired robustness specifications in terms of gain margin $(A_m)$ and phase margin $(\phi_{m})$. 
The highlight of the design is the choice of a fractional order (FO) filter in the IMC structure which has two  parameters ($\lambda$ and $\beta$) to tune as compared to only one tuning parameter ($\lambda$) for traditionally used integer order (IO) filter.
These parameters are evaluated for the controller, so that $A_m$ and $\phi_{m}$ can be chosen independently. 
A new methodology is proposed to find a complete solution for controller parameters, the methodology also gives the system gain cross-over frequency ($\omega_{g}$) and phase cross-over frequency ($\omega_{p}$). Moreover, the solution is found without any approximation of the delay term appearing in the controller.
\end{abstract}


\section{Introduction}
The IMC control strategy was proposed by M.Morari et al. in 1982 \cite{CEGarcia_1982_paper1, DRivera_1986_paper4}.
The IMC control structure as shown in Fig.\ref{figure_IMC_strucutre}(a) constitutes the inverse of the minimum phase part of the process model $\hat{G}_m(s)$ which is augmented with a filter $\Theta(s)$ and is given as $Q(s) = \hat{G}_m(s)/\Theta(s)$ \cite{MMorari_RobustPC_Book_1989}, such that $\mathop {\lim }\limits_{\omega  \to 0} \Theta (j\omega ) = 1$. Usually the filter is considered as $\Theta(s) = 1/(\lambda s+1)^r$ where $r$ is so chosen that $Q(s)$ becomes bi-proper. The parameter $\lambda$ is used to tune the controller to get the desired closed loop response of the system \cite{MMorari_RobustPC_Book_1989}. This is the standard IO  filter design problem for the IMC, where any value of $\lambda$ in the filter provides some $A_m$ and $\phi_m$ to the closed loop system, or in other words, for desired $A_m$ and $\phi_m$, parameter $\lambda$ is evaluated.

To the best of author's knowledge, only two major contributions \cite{IKaya_IMC_PI_GPM_2004, WKHo_IMC_PID_GPM_2001}  are present in literature which design IMC controller to satisfy $A_m$ and $\phi_{m}$ specifications simultaneously, whereas \cite{CWChu_IMCPID_Adaptive_2011} implemented the work in \cite{WKHo_IMC_PID_GPM_2001} in adaptive control setting. All three of them are IO-IMC for FOPTD processes.

The major limitation of these controllers is the presence of only one tuning parameter $\lambda$ which limits the domain of selection of desired $A_m$ and $\phi_{m}$. The achievable $A_m$ and $\phi_{m}$ are related through a mathematical expression which represents a curve (say, $A_m-\phi_{m}$ curve) in a 2-D \textcolor{black}{space} and desired $A_m$ and $\phi_{m}$ can be selected from that curve only. Also, the possible $\phi_m$ is restricted to  $(0,\pi/2)$. 

\textcolor{black}{To overcome these issues an FO filter $\Theta(s)=1/(\lambda s^{\beta}+1)$ is considered in this work which introduces an additional parameter $\beta$ into the controller. With two tuning parameters $(\lambda, \beta )$ instead of only $\lambda$ as in IO filter, the range of selection of desired $A_m$ and $\phi_{m}$ becomes a 2-D surface and hence it becomes possible to select $A_m$ and $\phi_{m}$ independently.}.

In this paper, a new solution method is developed to simultaneously satisfy $A_m \ge 2$ and $\phi_{m} \in(0,\pi)$ for $\theta/\tau > 0$, where $\theta$ is the delay and $\tau$ is the time constant of the plant, for FOPTD processes. The methodology also provides the solution of the system $\omega_{g}$ and $\omega_{p}$ considering them as the transitional variable.
It is also the first time that, the solution is attempted without any approximation of the delay term in the controller.

The contents of this manuscript are as follows: Section-II contains complete controller design methodology, derivations, proofs and controller design steps. Section-III, analysis on disturbance rejection with the proposed controller is given. The proposed methodology is verified with an example in Section-IV and Section-V is dedicated to discussion and conclusions. 
\section{Internal Model Controller}
\label{sec:main}
\begin{figure} 
	\centering
	\subfigure[]{\includegraphics[width=2.4in] {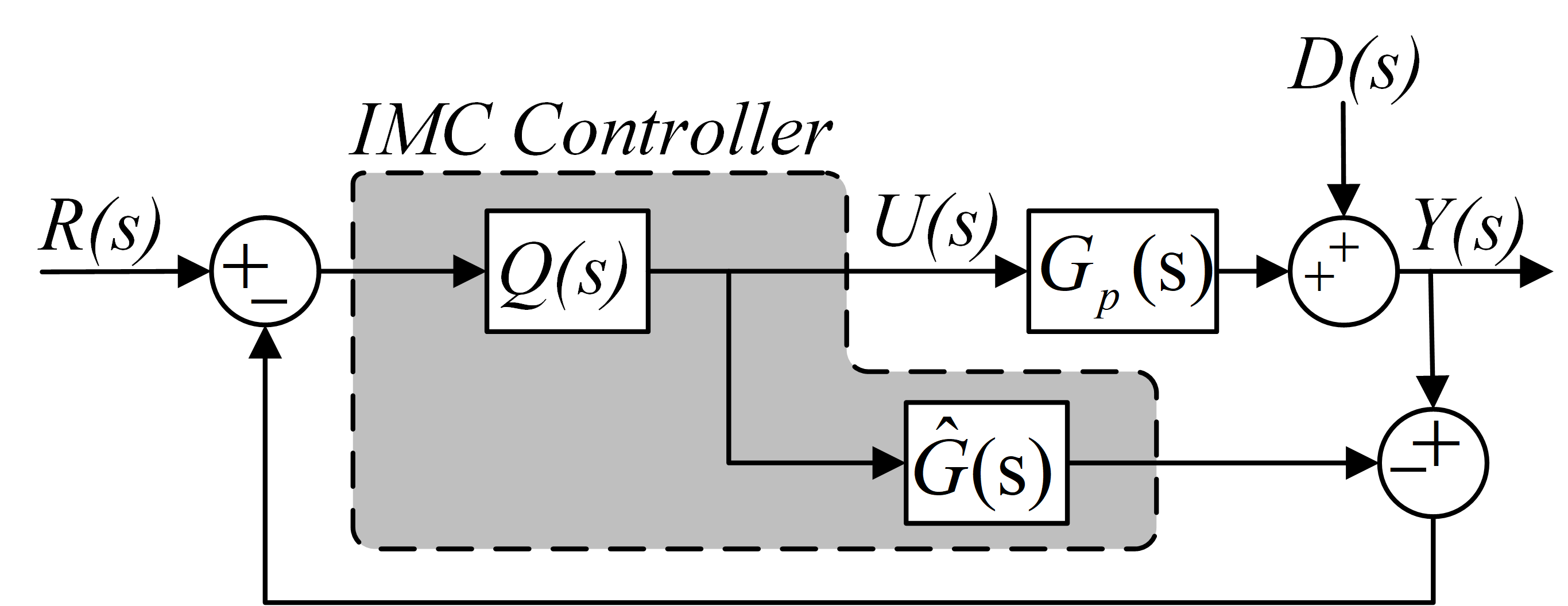}}
	\subfigure[]{\includegraphics[width=2.3in] {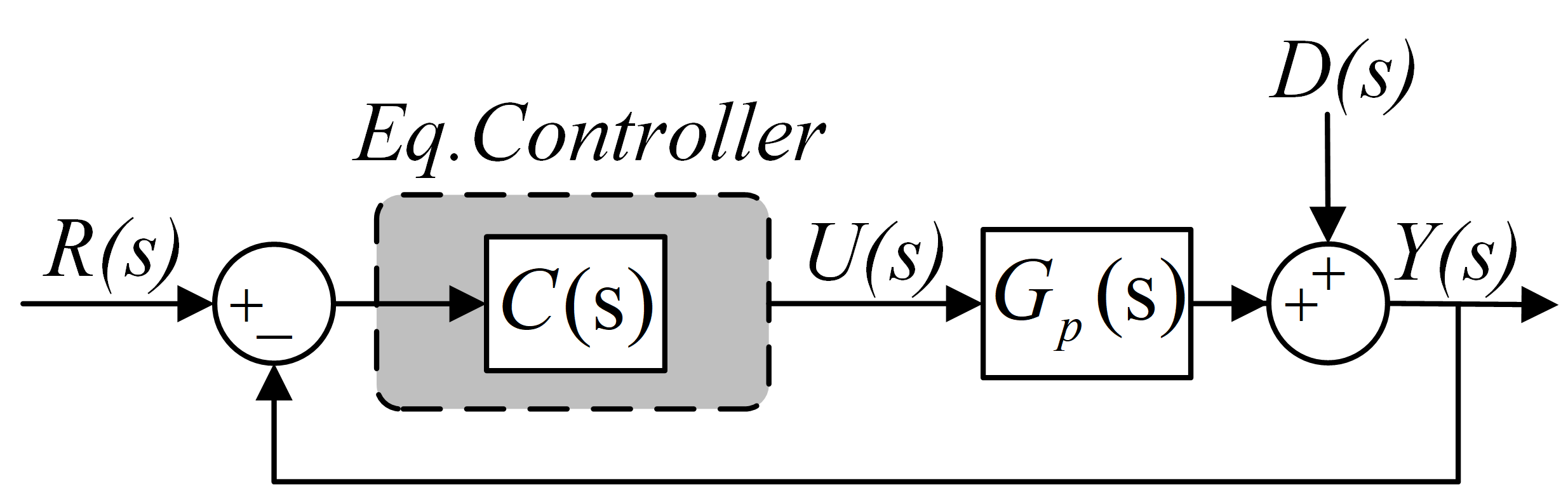}}
	\caption{(a) IMC control and (b) Equivalent classical control}
	\label{figure_IMC_strucutre}
\end{figure}
In Fig.1(a), $Q(s)$ is the IMC controller, ${G_p}(s)$ is the process and  $\hat G(s)$ is the model of the process used in the control loop. In Fig.\ref{figure_IMC_strucutre}(b), $C(s)$ is the equivalent classical controller obtained by block diagram reduction in the structure of  Fig.\ref{figure_IMC_strucutre}(a) and is given as $C(s) = Q(s)/(1-Q(s)\hat{G}(s))$.
Assuming plant behavior as FOPTD and exact modeling as per \cite{MMorari_RobustPC_Book_1989}, $\hat{G}(s)$ can be written as
\begin{equation} \label{Gps_equal_hat_Gs}
\hat{G}(s) =G_p(s) =  \frac{k}{\tau s+1} e^{-\theta s}
\end{equation}
Segregating the minimum phase (MP) and non-minimum phase (NMP) part of the process and the model, we have
$G_p(s)=G_{pm}(s)G_{pa}(s)$ and $\hat{G}(s) = \hat{G}_m(s) \hat{G}_a(s)$,
where subscript $'a'$ represents NMP or all-pass part and subscript $'m'$ represents MP part of the transfer function. Since $\hat{G}(s)=G_p(s)$ therefore,
\begin{equation} \label{G_hat_a=k/tau+1}
\hat{G}_{m}(s) = G_{pm}(s) = \frac{k}{\tau s + 1}; \qquad \hat{G}_{a}(s) = G_{pa}(s) = e^{-\theta s}
\end{equation} 
The IMC controller $Q(s)$ is given as \cite{MMorari_RobustPC_Book_1989}
\begin{equation} \label{Q=f/hatGm}
Q(s) = \frac{\Theta(s)}{{\hat G_m}(s)}
\end{equation}
where $\Theta(s)$ is an FO filter chosen in such a way that $Q(s)$ is realizable 
and  $ \mathop {\lim }\limits_{s \to 0} {\Theta}(s) = 1 $
\cite{MMorari_RobustPC_Book_1989}.
The FO filter considered in this paper is 
\begin{equation} \label{IMC_filter_Theta}
\Theta(s) = \frac{1}{{\lambda {s^\beta } + 1}}; \quad \beta  \in (0,2), \, \lambda > 0
\end{equation}
\textcolor{black}{where $\beta$ denoting the fractional order is an additional parameter along with the filter constant $\lambda$.}
\subsection{$A_m$ and $\phi_{m}$ specifications:}
Substituting $\hat{G}_m(s)$ from (\ref{G_hat_a=k/tau+1}) and $\Theta(s)$ from  (\ref{IMC_filter_Theta}) in
(\ref{Q=f/hatGm}), we get
\begin{equation} \label{Qs_final}
Q(s) = \frac{\tau s +1}{k(\lambda s^\beta +1)}
\end{equation}
The equivalent controller in Fig.1(b) can be obtained by substituting \textcolor{black}{$Q(s)$ from} (\ref{Qs_final}) and $\hat{G}(s)$ from (\ref{Gps_equal_hat_Gs}) in 
$C(s)= Q(s)/(1-Q(s)\hat{G}(s))$, and we get
\begin{equation} \label{Cs}
C(s) = \frac{1}{{\lambda {s^\beta } + 1 - {e^{ - \theta s}}}} \cdot \frac{1}{{{{\hat G}_m}(s)}}
\end{equation}
The open loop transfer function based on Fig.1(b) is $L(s) = C(s)G_p(s)$ considering $D(s) = 0$, since $D(s)$ has no influence on the 
$A_m$ and $\phi_{m}$ specifications.
Substituting $C(s)$ from (\ref{Cs}) and $G_p(s)$ from (\ref{Gps_equal_hat_Gs}), the open loop transfer function (OLTF) becomes
\begin{equation} \label{OL_TF}
L(s) = \frac{{{e^{ - \theta s}}}}{{\lambda {s^\beta } +1 - {e^{ - \theta s}}}}
\end{equation}
For an OLTF $L(s)$, the $A_m$ and $\phi_{m}$ specifications are given as
\begin{align} \label{PM_condition_1}
{\left. {L(j\omega )} \right|_{\omega  = {\omega _g}}} =  - {e^{j{\phi _m}}} \\
\label{GM_condition_2}
\qquad
{\left. {L(j\omega )} \right|_{\omega  = {\omega _p}}} = \frac{{ - 1}}{{{A_m}}}
\end{align}
where $\omega_g$ is gain cross-over frequency and $\omega_p$ is phase cross-over frequency of the closed loop system. To have desired $A_m$ and $\phi_m$, we need to find $\lambda$ and $\beta$ which  satisfy (\ref{PM_condition_1}) and (\ref{GM_condition_2}) simultaneously.  
Substituting $s = j \omega$ in (\ref{OL_TF}), we get
\begin{align} \label{Ol_TF_jw}
L(j\omega ) = \frac{{{e^{ - j\theta \omega }}}}{\begin{array}{c}
	\left( {1 + \lambda {\omega ^\beta }\cos \left( {\frac{{\beta \pi }}{2}} \right) - \cos \left( {\theta \omega } \right)} \right) + 
	j\left( {\lambda {\omega ^\beta }\sin \left( {\frac{{\beta \pi }}{2}} \right) + \sin \left( {\theta \omega } \right)} \right)
	\end{array}}
\end{align}
Form (\ref{Ol_TF_jw}) and (\ref{PM_condition_1}), we have
\begin{align*}
\begin{array}{*{20}{c}}
\begin{array}{c}
\Rightarrow \left( {1 + \lambda \omega _g^\beta \cos \left( {\frac{{\beta \pi }}{2}} \right) - \cos \left( {\theta {\omega _g}} \right)} \right) + 
j\left( {\lambda \omega _g^\beta \sin \left( {\frac{{\beta \pi }}{2}} \right) + \sin \left( {\theta {\omega _g}} \right)} \right)
\end{array}&{ =  - {e^{ - j\left( {{\phi _m} + \theta {\omega _g}} \right)}}}
\end{array}
\end{align*}
\begin{align} \label{L_spec1_in_real_imag}
\begin{array}{c}
\Rightarrow \left( {1 + \lambda \omega _g^\beta \cos \left( {\frac{{\beta \pi }}{2}} \right) - \cos \left( {\theta {\omega _g}} \right)} \right) + 
j\left( {\lambda \omega _g^\beta \sin \left( {\frac{{\beta \pi }}{2}} \right) + \sin \left( {\theta {\omega _g}} \right)} \right)
=  - \cos \left( {{\phi _m} + \theta {\omega _g}} \right) + j\sin \left( {{\phi _m} + \theta {\omega _g}} \right)
\end{array}
\end{align}
Equating real and imaginary part in (\ref{L_spec1_in_real_imag}), we get
\begin{align} \label{PM_condition_eq1_tosolve}
{1 + \lambda \omega _g^\beta \cos \left( {\frac{{\beta \pi }}{2}} \right) - \cos \left( {\theta {\omega _g}} \right)}  =  - \cos \left( {{\phi _m} + \theta {\omega _g}} \right) \\
\label{PM_condition_eq2_tosolve}
\text{and} \qquad
\lambda \omega _g^\beta \sin \left( {\frac{{\beta \pi }}{2}} \right) + \sin \left( {\theta {\omega _g}} \right) = \sin \left( {{\phi _m} + \theta {\omega _g}} \right)
\end{align}
Similarly, from (\ref{Ol_TF_jw}) and (\ref{GM_condition_2}) we get
\[\begin{array}{*{20}{c}}
\left( {1 + \lambda \omega _p^\beta \cos \left( {\frac{{\beta \pi }}{2}} \right) - \cos \left( {\theta {\omega _p}} \right)} \right) + 
j\left( {\lambda \omega _p^\beta \sin \left( {\frac{{\beta \pi }}{2}} \right) + \sin \left( {\theta {\omega _p}} \right)} \right) 
 =  - {A_m}{e^{ - j\theta {\omega _p}}}
\end{array}\]
\begin{align} \label{L_spec2_in_real_imag}
\Rightarrow \left( {1 + \lambda \omega _p^\beta \cos \left( {\frac{{\beta \pi }}{2}} \right) - \cos \left( {\theta {\omega _p}} \right)} \right)
+ j\left( {\lambda \omega _p^\beta \sin \left( {\frac{{\beta \pi }}{2}} \right) + \sin \left( {\theta {\omega _p}} \right)} \right)
=  - {A_m}\cos \left( {\theta {\omega _p}} \right)
+ j{A_m}\sin \left( {\theta {\omega _p}} \right)
\end{align}
Equating real and imaginary part in (\ref{L_spec2_in_real_imag}), we get
\begin{align} \label{GM_condition_eq3_tosolve}
1 + \lambda \omega _p^\beta \cos \left( {\frac{{\beta \pi }}{2}} \right) - \cos \left( {\theta {\omega _p}} \right) =  - {A_m}\cos \left( {\theta {\omega _p}} \right) \\
\label{GM_condition_eq4_tosolve}
\text{and} \qquad
{\lambda \omega _p^\beta \sin \left( {\frac{{\beta \pi }}{2}} \right) + \sin \left( {\theta {\omega _p}} \right)} = {A_m}\sin \left( {\theta {\omega _p}} \right)
\end{align}

The problem has four non-linear transcendental equations (\ref{PM_condition_eq1_tosolve}), (\ref{PM_condition_eq2_tosolve}), (\ref{GM_condition_eq3_tosolve}), and (\ref{GM_condition_eq4_tosolve}) with four unknowns $\lambda$, $\beta$, $\omega_g$ \& $\omega_p$, where $\lambda$ and $\beta$ are to be evaluated such that the desired $\phi_m$ and $A_m$ are satisfied simultaneously, $\omega_g$ and $\omega_p$ begin transitional variables. 
The solution of equations (\ref{PM_condition_eq1_tosolve}) and (\ref{PM_condition_eq2_tosolve}) shall give $\lambda$ and $\beta$ such that it satisfies (\ref{PM_condition_1}) or in other words it satisfies $\phi_m$ and $\omega_g$ simultaneously. 
Similarly, solution of equations (\ref{GM_condition_eq3_tosolve}) and (\ref{GM_condition_eq4_tosolve}) shall give the values of $\lambda$ and $\beta$ such that it satisfies (\ref{GM_condition_2}) or equivalently, it satisfies $A_m$ and $\omega_{p}$ simultaneously. 

\subsection{Design Philosophy:}
\textcolor{black}{
	Let us denote $\lambda=\lambda_a$ in (\ref{L_spec1_in_real_imag}). Then $\lambda$ in (\ref{PM_condition_eq1_tosolve}) and (\ref{PM_condition_eq2_tosolve}) becomes $\lambda_a$ as they come from the same equation  (\ref{L_spec1_in_real_imag}). Similarly, denoting $\lambda = \lambda_b$ in (\ref{L_spec2_in_real_imag}), $\lambda$ in (\ref{GM_condition_eq3_tosolve}) and (\ref{GM_condition_eq4_tosolve}) becomes $\lambda_b$.
}

Then the solution of (\ref{PM_condition_eq1_tosolve}) and  (\ref{PM_condition_eq2_tosolve}) would give a set $\Delta_a$, containing $\{ \lambda_a, \beta_{\omega_g} \}$ corresponding
to a set $\Xi_a$ which contains those $\{\phi_{m}, \omega_{g}\}$ which
satisfy (\ref{L_spec1_in_real_imag}). Similarly, the solution of (\ref{GM_condition_eq3_tosolve}) and (\ref{GM_condition_eq4_tosolve}) will give a set $\Delta_b$, containing $\{  \lambda_b, \beta_{\omega_p} \}$ corresponding to a set $\Xi_b$ which contains those $\{A_m, \omega_p \}$ which 
satisfy (\ref{L_spec2_in_real_imag}). 
Then the intersection set $\Delta_c = \Delta_a \cap  \Delta_b$ shall contain $\{ \lambda_c^{\star},\beta_c^{\star} \}$ which satisfy both (\ref{L_spec1_in_real_imag}) and (\ref{L_spec2_in_real_imag}),
for a given $\{A_m^{\star}, \phi_m^{\star}\}$, where $A_m^{\star}$ and $\phi_m^{\star}$ are elements from the set $\Xi_c = \Xi_a \cap \Xi_b$ associated with some corresponding $\omega_{p}^{*}$ and $\omega_{g}^{*}$. 

To obtain the solution, first we assume $\beta_{\omega_{g}}=\beta\in(0,2)$ and find corresponding $\omega_{g}$ in terms of $\beta_{{\omega _g} }$ from (\ref{PM_condition_eq1_tosolve}) and (\ref{PM_condition_eq2_tosolve}) by eliminating $\lambda_a $. Substituting this $\{ \omega_{g}, \beta_{{\omega _g}} \}$ in (\ref{PM_condition_eq1_tosolve}) or (\ref{PM_condition_eq2_tosolve}), we get corresponding $\lambda_a$ values. Similarly we consider $\beta_{\omega_{p}}=\beta\in(0,2)$ and find $\omega_p$ in terms of $\beta_{{\omega _p}}$ from (\ref{GM_condition_eq3_tosolve}) and (\ref{GM_condition_eq4_tosolve}) by eliminating $\lambda_b$. 
Substituting this $\{ \omega_{p}, \beta_{{\omega _p}} \}$ in (\ref{GM_condition_eq3_tosolve}) or (\ref{GM_condition_eq4_tosolve}), we get corresponding $\lambda_b$.
Therefore,  $\{  \lambda_a, \beta_{\omega_{g}} \}$ lead to a set $\Xi_a$ which contains all $\{\phi_m, \omega_g\}$ which  $\{ \lambda_a, \beta_{\omega_{g}} \}$ can satisfy. Similarly, $\{ \lambda_b,\beta_{{\omega _p}} \}$ lead to a set $\Xi_b$ which contains all $\{A_m, \omega_{p}\}$ which $\{ \lambda_b,\beta_{\omega_{p}} \}$ can satisfy. 
Then $\beta_{\omega_{g}} vs. \lambda_a$ and $\beta_{\omega_{p}} vs. \lambda_b$ are plotted together to find the intersection of the curve which gives $\lambda^{\star}$ and $\beta^{\star}$ which satisfies $\phi_{m}^{\star}$ at some $\omega_{g}^{\star}$ and $A_m^{\star}$ at some $\omega_p^{\star}$.

\subsection{Finding $\omega_{g} \in \Xi_a$:}
Eliminating $\lambda$ from (\ref{PM_condition_eq1_tosolve}) and (\ref{PM_condition_eq2_tosolve}), we get
\begin{equation} \label{wg_from_PM_condn}
\tan \left( {\frac{{\beta \pi }}{2}} \right) = \frac{{\sin \left( {{\phi _m} + \theta {\omega _g}} \right) - \sin \left( {\theta {\omega _g}} \right)}}{{\cos \left( {\theta {\omega _g}} \right) - \cos \left( {{\phi _m} + \theta {\omega _g}} \right) - 1}}
\end{equation}

In (\ref{wg_from_PM_condn}), $\beta$ and $\omega_{g}$ are unknown, whereas \textcolor{black}{$\phi_{m}$ is given as the desired phase margin} \textcolor{black}{and $\theta$  is given from  the system model.} However, \textcolor{black}{since} the range of $\beta$ is \textcolor{black}{fixed}, therefore, $\omega_g$ can be found in terms of $\beta$.

Cross multiplying with numerator and denominator terms in LHS and RHS of the equation in (\ref{wg_from_PM_condn}), we get
\begin{align} \label{wg_from_PM_condn_simplify_1}
\begin{array}{c}
\tan \left( {\frac{{\beta \pi }}{2}} \right)\cos \left( {\theta {\omega _g}} \right) - \tan \left( {\frac{{\beta \pi }}{2}} \right)\cos \left( {{\phi _m} + \theta {\omega _g}} \right)  - \tan \left( {\frac{{\beta \pi }}{2}} \right)  = 
\sin \left( {{\phi _m} + \theta {\omega _g}} \right) - \sin \left( {\theta {\omega _g}} \right)
\end{array}
\end{align}
Using trigonometric identities
and simplifying, we can write (\ref{wg_from_PM_condn_simplify_1}) as
\begin{align}  \label{wg_equation}
a_1\cos \left( {\theta {\omega _g}} \right) + b_1\sin \left( {\theta {\omega _g}} \right) = c_1
\end{align}
where
\begin{align}  \label{a1b1c1_initial}
\begin{array}{*{20}{c}}
{{a_1} = \tan \left( {\frac{{\beta \pi }}{2}} \right)\left( {1 - \cos \left( {{\phi _m}} \right)} \right) - \sin \left( {{\phi _m}} \right)} ; \, 
{{b_1} = \tan \left( {\frac{{\beta \pi }}{2}} \right)\sin \left( {{\phi _m}} \right) + \left( {1 - \cos \left( {{\phi _m}} \right)} \right)}; \,
{{c_1} = \tan \left( {\frac{{\beta \pi }}{2}} \right)}
\end{array}
\end{align}
In (\ref{wg_equation}), let $a_1 = r_1\cos(\alpha_1)$ \& $b_1 = r_1\sin(\alpha_1)$ where $\alpha_1 \in \Re$ and $\alpha_1$ is in radians. Then, $r_1=\sqrt{a_1^{2}+b_1^{2}}$ and $\alpha_1 = \tan^{-1}\left( \frac{b_1}{a_1}\right)$. This transforms (\ref{wg_equation}) as
\begin{equation} \label{omega_g_initial}
\cos \left( {\alpha_1  - \theta {\omega _g}} \right) = \frac{c_1}{r_1}
\end{equation}
Note that $\alpha_1$, $c_1$ and $r_1$ ultimately depend only on $\beta$.
\begin{lemma}	[Simplification for $c_1/r_1$] 
	If $a_1$, $b_1$ and $c_1$ are as given in (\ref{a1b1c1_initial}), then 
	\begin{equation} \label{c_1/r_1_final}
	\frac{c_1}{r_1} = \frac{{\sin \left( {\frac{{\beta \pi }}{2}} \right)}}{{2\sin \left( {\frac{{{\phi _m}}}{2}} \right)}}; \quad  \beta  \in \left( {0,2} \right)
	\end{equation}
\end{lemma}
\begin{proof}
	Substituting $a_1$ and $b_1$ from (\ref{a1b1c1_initial}), \textcolor{black}{and further simplification of the trigonometric terms, we get}
	\begin{align}
	{a_1^2} + {b_1^2} = {\left( {\frac{{2\sin \left( {\frac{{{\phi _m}}}{2}} \right)}}{{\cos \left( {\frac{{\beta \pi }}{2}} \right)}}} \right)^2}
	\end{align}
	Hence $r_1=\sqrt{a_1^2+b_1^2}= \frac{{2\sin \left( {\frac{{{\phi _m}}}{2}} \right)}}{{\cos \left( {\frac{{\beta \pi }}{2}} \right)}}$.
	\textcolor{black}{Then $	\frac{c_1}{r_1} = \frac{{\sin \left( {\frac{{\beta \pi }}{2}} \right)}}{{2\sin \left( {\frac{{{\phi _m}}}{2}} \right)}}$.}
\end{proof}
From (\ref{c_1/r_1_final}), it is clear that $\phi_m=0$ is not possible as at $\phi_m=0$, $c_1/r_1\to\infty$ and for such case the solution of $\omega_g$ from (\ref{omega_g_initial}) does not exist.
\begin{lemma}[Simplification for $\alpha_1$]
	If $a_1$ and $b_1$ are as given in (\ref{a1b1c1_initial}), then
	\begin{equation} \label{alpha1_final}
	\alpha_1  = \left\{ {\begin{array}{*{20}{c}}
		{ - \left( {\frac{{\beta \pi }}{2} + \frac{{{\phi _m}}}{2}} \right)}& \textcolor{black}{{\frac{{\beta \pi }}{2} + \frac{{{\phi _m}}}{2} \in \left( {0,\frac{\pi }{2}} \right)}}\\
		{\pi  - \left( {\frac{{\beta \pi }}{2} + \frac{{{\phi _m}}}{2}} \right)}&{\frac{{\beta \pi }}{2} + \frac{{{\phi _m}}}{2} \in \left( {\frac{\pi }{2},\frac{{3\pi }}{2}} \right)}
		\end{array}} \right.
	\end{equation}
\end{lemma}
\begin{proof}
	Since $\alpha_1  = \tan ^{ - 1} {b_1}/{a_1}$, substituting $a_1$ and $b_1$ from (\ref{a1b1c1_initial}), 
	\textcolor{black}{and using trigonometric identities for simplification, we get}
	$$ \alpha_1  = {\tan ^{ - 1}}\left( {\frac{{\tan \left( {\frac{{\beta \pi }}{2}} \right) + \tan \left( {\frac{{{\phi _m}}}{2}} \right)}}{{\tan \left( {\frac{{\beta \pi }}{2}} \right)\tan \left( {\frac{{{\phi _m}}}{2}} \right) - 1}}} \right)$$
	Therefore,
	\begin{equation} \label{alpha_in_atan_tan}
	\alpha_1  = {\tan ^{ - 1}}\left( { - \tan \left( {\frac{{\beta \pi }}{2} + \frac{{{\phi _m}}}{2}} \right)} \right) =  - {\tan ^{ - 1}}\left( {\tan \left( {\frac{{\beta \pi }}{2} + \frac{{{\phi _m}}}{2}} \right)} \right)
	\end{equation}
	From fundamentals of inverse trigonometry, ${\tan ^{ - 1}}\left( {\tan \left( x  \right)} \right) = x$, only if
	$x \in (-\pi/2, \pi/2)$. If $x$ lies outside this range, the origin needs to be shifted to the desired domain of the argument $x$ to get the correct result. Therefore, we obtain
	\begin{equation*} 
	\alpha_1  = \left\{ {\begin{array}{*{20}{c}}
		{ - \left( {\frac{{\beta \pi }}{2} + \frac{{{\phi _m}}}{2}} \right)}&{\frac{{\beta \pi }}{2} + \frac{{{\phi _m}}}{2} \in \left( {0,\frac{\pi }{2}} \right)}\\
		{\pi  - \left( {\frac{{\beta \pi }}{2} + \frac{{{\phi _m}}}{2}} \right)}&{\frac{{\beta \pi }}{2} + \frac{{{\phi _m}}}{2} \in \left( {\frac{\pi }{2},\frac{{3\pi }}{2}} \right)}
		\end{array}} \right.
	\end{equation*}
\end{proof}
\begin{theorem} [Solution of $\omega_{g}$]
	The gain cross-over frequency $\omega_{g}$ which satisfies (\ref{PM_condition_1}) can be found as
	\begin{equation} \label{wg_sol_final_expanded}
	{\omega _g} = \left\{ {\begin{array}{*{20}{c}}
		{\frac{1}{\theta }\left( { - {{\cos }^{ - 1}}\left( {\frac{{\sin \left( {\frac{{\beta \pi }}{2}} \right)}}{{2\sin \left( {\frac{{{\phi _m}}}{2}} \right)}}} \right) - \left( {\frac{{\beta \pi }}{2} + \frac{{{\phi _m}}}{2}} \right)} \right)} \\ 
		\textcolor{black}{{;\;\beta  \in \left( {0,{\beta _{x1}}} \right)}} \\ 
		{\frac{1}{\theta }\left( { - {{\cos }^{ - 1}}\left( {\frac{{\sin \left( {\frac{{\beta \pi }}{2}} \right)}}{{2\sin \left( {\frac{{{\phi _m}}}{2}} \right)}}} \right) + \pi  - \left( {\frac{{\beta \pi }}{2} + \frac{{{\phi _m}}}{2}} \right)} \right)} \\ 
		{;\;\beta  \in \left( {{\beta _{x1}},{\beta _{x3}}} \right)} 
		\end{array}} \right.
	\end{equation}
	where ${{\beta _{x1}}} = (\pi-\phi_{m})/\pi$ and ${{\beta _{x3}}} = (3\pi-\phi_{m})/\pi$.
\end{theorem}
\begin{proof}
	From (\ref{omega_g_initial}), $\omega_{g}$ becomes
	\begin{equation} \label{wg_sol_final}
	{\omega _g} = \frac{1}{\theta }\left\{ { - {{\cos }^{ - 1}}\left( {\frac{{{c_1}}}{{{r_1}}}} \right) + {\alpha _1}} \right\}
	\end{equation}
	Substituting $\cos^{-1}(c_1/r_1)$ from (\ref{c_1/r_1_final}) and  $\alpha_1$ from (\ref{alpha1_final}) in (\ref{wg_sol_final}), $\omega_{g}$ is given as
	\begin{equation} \label{wg_sol_final_expanded_1}
	{\omega _g} = \left\{ {\begin{array}{*{20}{c}}
		{\frac{1}{\theta }\left( { - {{\cos }^{ - 1}}\left( {\frac{{\sin \left( {\frac{{\beta \pi }}{2}} \right)}}{{2\sin \left( {\frac{{{\phi _m}}}{2}} \right)}}} \right) - \left( {\frac{{\beta \pi }}{2} + \frac{{{\phi _m}}}{2}} \right)} \right)}\\
		\textcolor{black}{{;\;\frac{{\beta \pi }}{2} + \frac{{{\phi _m}}}{2} \in \left( {0,\frac{\pi }{2}} \right)}}\\
		{\frac{1}{\theta }\left( { - {{\cos }^{ - 1}}\left( {\frac{{\sin \left( {\frac{{\beta \pi }}{2}} \right)}}{{2\sin \left( {\frac{{{\phi _m}}}{2}} \right)}}} \right) + \pi  - \left( {\frac{{\beta \pi }}{2} + \frac{{{\phi _m}}}{2}} \right)} \right)}\\
		{;\;\frac{{\beta \pi }}{2} + \frac{{{\phi _m}}}{2} \in \left( {\frac{\pi }{2},\frac{{3\pi }}{2}} \right)}
		\end{array}} \right.
	\end{equation}
	Simplifying the domain in (\ref{wg_sol_final_expanded_1}) for $\beta$, we get
	$\beta \in \textcolor{black}{(-\phi_{m}/\pi , (\pi-\phi_{m})/\pi )}$ for the first case and $\beta \in ((\pi-\phi_{m})/\pi, (3\pi-\phi_{m})/\pi )$ for the second case. However, $\beta\in(0,2)$ therefore, the lower bound for the first case will become zero. Therefore, the first \textcolor{black}{domain} will become $\beta \in \textcolor{black}{(0 , \beta_{x1})}$ and the second \textcolor{black}{domain} remains as it is, i.e., $\beta\in(\beta_{x1}, \beta_{x3})$.
\end{proof}
\subsection{Finding $\beta_{\omega_{g}}$ from set $\Xi_a$:}
In practice, control systems  always have real and positive gain cross-over frequency.
Let \textcolor{black}{${\rm B}_{\omega_{g}}$ be the set of} $\beta$ which gives some desired $\omega_{g}$ which is real and positive, i.e, ${{\rm B} _{\omega_{g}}} = \left\{ {\beta :{\omega _g} \in {\Re ^ + }} \right\}$. 

From (\ref{wg_sol_final_expanded_1}), the solution of $\omega_{g}$ depends on $\beta$ and $\phi_m$. However, $\phi_m$ is \textcolor{black}{given by} design specification, therefore, $\beta$ can be found in terms of $\phi_m$ such that $\omega_g \in \Re^+$.
This can be \textcolor{black}{done} in two steps. First, $\beta$ is obtained so that $\omega_g \in \Re$\textcolor{black}{, i.e.,  set} $ {{\rm B} _{\omega g\Re }} = \left\{ {\beta :{\omega _g} \in \Re } \right\} $ is evaluated.
\textcolor{black}{Thereafter,} set $ {{\rm B} _{\omega g + }} = \left\{ {\beta :{\omega _g} > 0} \right\} $ is obtained. Finally, $ {{\rm B} _{\omega g}} = {{\rm B} _{\omega g\Re}} \cap {{\rm B} _{\omega g + }} $. 
\begin{theorem}[ 
	\textcolor{black}{Finding $\rm B_{\omega_{g}\Re}$}]
	If $c_1/r_1$ is as given in (\ref{c_1/r_1_final}), then the solution of $\beta$ such that $\omega_{g}\in\Re$ is given by:\\
	(a)	 ${{\rm B}_{{\omega _g}\Re }} = \left\{ {\beta :\beta  \in (0,{\beta _{{\omega _g}\Re 1}}) \cup ({\beta _{{\omega _g}\Re 2}},2)} \right\}$ \textcolor{black}{when} $\phi_{m}\in(0,\pi/3)$ and \\
	(b)  \textcolor{black}{${{\rm B}_{{\omega _g}\Re }} = \left\{ {\beta :\beta  \in (0,2)} \right\}$} \textcolor{black}{when} $\phi_{m} \in [\pi/3, \pi)$, where
	\begin{align} \label{beta_wg_max_01}
	{\beta _{{\omega _g}\Re 1}} = \frac{2}{\pi }{\sin ^{ - 1}}\left( {2\sin \left( {\frac{{{\phi _m}}}{2}} \right)} \right)\\
	\label{beta_wg_Re_2_final}
	and \qquad
	{{\beta _{{\omega _g}\Re 2}} = \frac{2}{\pi }\left( {\pi  - {{\sin }^{ - 1}}\left( {2\sin \left( {\frac{{{\phi _m}}}{2}} \right)} \right)} \right)}
	\end{align}
\end{theorem}
\begin{proof}
	From (\ref{wg_sol_final}), $\omega_g \in \Re$ if $\left| c_1/r_1 \right| \le 1$, where
	$c_1/r_1$ is given in (\ref{c_1/r_1_final}). Thus,
	\begin{equation} \label{c_1/r_1_exp_final_1}
	\left| {\frac{{\sin \left( {\frac{{\beta \pi }}{2}} \right)}}{{2\sin \left( {\frac{{{\phi _m}}}{2}} \right)}}} \right| \le 1
	\end{equation}
	Since $c_1/r_1>0\,\forall\,\beta\in(0,2)$ and $\phi_{m}\in(0,\pi)$, therefore, modulus sign can be \textcolor{black}{removed} from (\ref{c_1/r_1_exp_final_1}), and we get
	\begin{equation} \label{c_1/r_1_exp_final}
	\sin \left( {\frac{{\beta \pi }}{2}} \right) \le 2\sin \left( {\frac{{{\phi _m}}}{2}} \right)
	\end{equation} 
	We need to find $\beta$ such that (\ref{c_1/r_1_exp_final}) holds.
	\begin{figure}  
		\centering
		\includegraphics[width=2in] {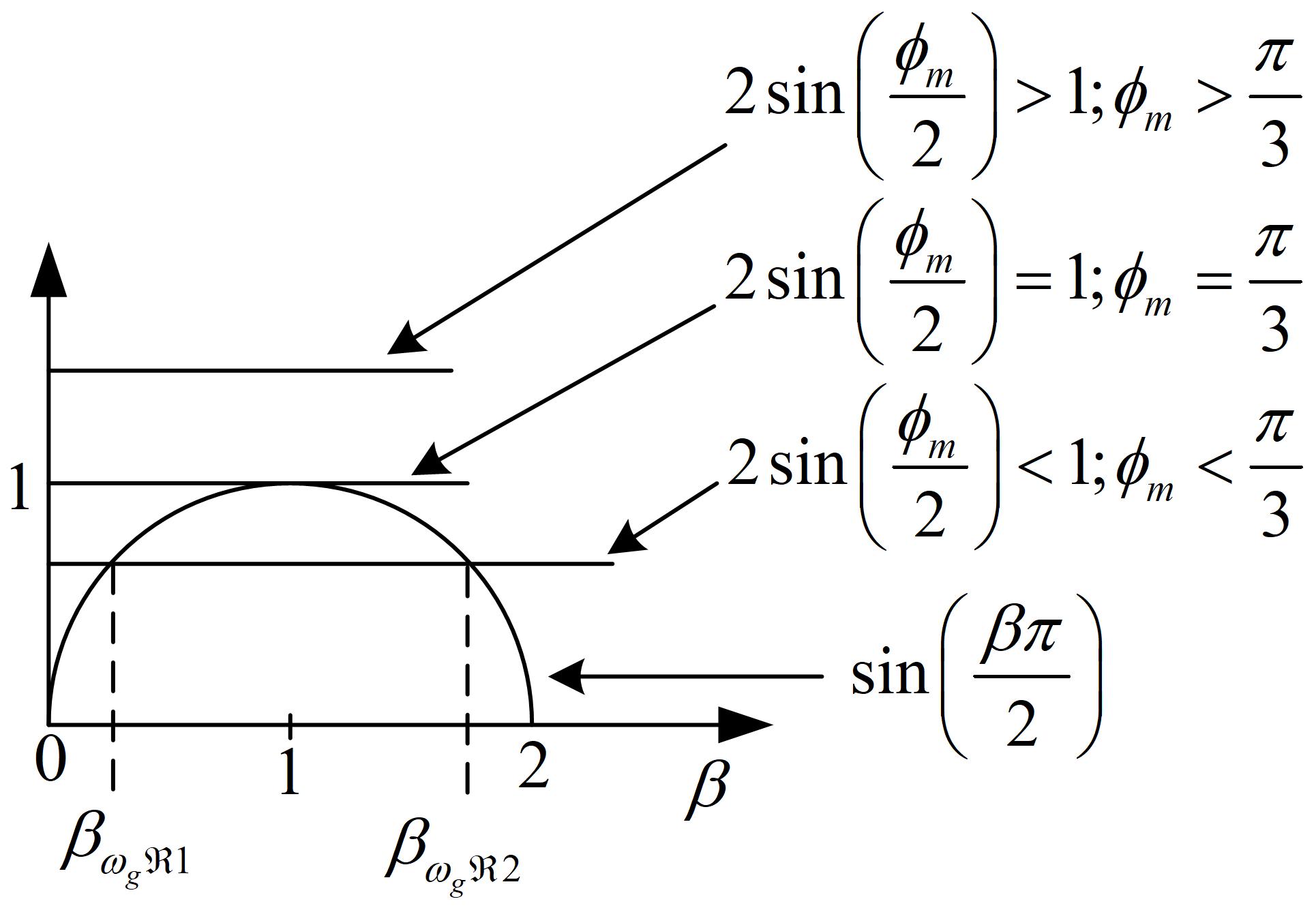}
		\caption{Graphical representation of (\ref{c_1/r_1_exp_final})}
		\label{fig_c_1/r_1_graphical}
	\end{figure}
	In (\ref{c_1/r_1_exp_final}), $\beta$ can be considered as a variable and $\phi_{m}$ as an arbitrary constant given by design specification.
	\textcolor{black}{ Plotting LHS and RHS of (\ref{c_1/r_1_exp_final}) in}
	Fig.\ref{fig_c_1/r_1_graphical}, two cutoff points of $\beta$ arise, namely $\beta_{{\omega _g}\Re 1}$ and $\beta_{{\omega _g}\Re 2}$ when $\phi_{m}\in(0,\pi/3)$. 
	At $\phi_m=\pi/3$, ${\beta _{{\omega _g}\Re 1}} = {\beta _{{\omega _g}\Re 2}} = 1$ and for $\phi_{m}\in[\pi/3,\pi)$, we have $\beta\in(0,2)$.
	Therefore, if $\phi_{m}\in(0,\pi/3)$, $\beta \in (0,{\beta _{{\omega _g}\Re 1}}) \cup ({\beta _{{\omega _g}\Re 2}},2)$ .
	The boundary values ${\beta _{{\omega _g}\Re 1}}$ and ${\beta _{{\omega _g}\Re 2}}$ can be found by equating LHS and RHS in (\ref{c_1/r_1_exp_final}). 
	i.e
	\begin{equation} \label{beta_wg_R1_R2_equation}
	\sin \left( {\frac{{\beta \pi }}{2}} \right) = 2\sin \left( {\frac{{{\phi _m}}}{2}} \right)
	\end{equation}
	The function $\arcsin$ is defined for the fundamental period $[-\pi/2,\pi/2]$, i.e, if $y\,\in\,[-1,1]$ then $x = {\arcsin}\left( y \right)$ \textcolor{black}{implies} $x \in [-\pi/2, \pi/2]$. 
	In the present context, in (\ref{beta_wg_R1_R2_equation}), if the argument is in LHS, for $\beta\in(0,1]$ we have $\beta \pi / 2 \in (0,\pi/2]$ which lies in the fundamental period. Therefore taking inverse will give the correct solution.  
	Hence, if $\beta\in(0,1]$, (\ref{beta_wg_R1_R2_equation}) can be evaluated by directly taking $\arcsin$, and we get
	\begin{equation}
	{\beta _{{\omega _g}\Re 1}} = \frac{2}{\pi }{\sin ^{ - 1}}\left( {2\sin \left( {\frac{{{\phi _m}}}{2}} \right)} \right)
	\end{equation}
	\textcolor{black}{However}, in (\ref{beta_wg_R1_R2_equation}), for  $\beta\in(1,2)$, we have $\beta \pi/2 \in (\pi/2,\pi)$, which lies outside of the fundamental period. Therefore, taking $\arcsin$ \textcolor{black}{directly} the solution will lie in $[-\pi/2,\pi/2]$, which is incorrect. In such cases, the solution can be corrected by shifting the origin of inverse function in the appropriate domain of the argument. 
	Therefore, for $\beta\in(1,2)$, the origin of the $\arcsin$ need to be shifted
	as below 
	\begin{equation} \label{beta_wg_min_12_initial}
	\frac{{\beta \pi }}{2} - \pi  =  - \arcsin \left( {2\sin \left( {\frac{{{\phi _m}}}{2}} \right)} \right)
	\end{equation}
	\textcolor{black}{which gives}
	\begin{equation}
	{\beta _{{\omega _g}\Re 2}} = \frac{2}{\pi }\left( {\pi  - {{\sin }^{ - 1}}\left( {2\sin \left( {\frac{{{\phi _m}}}{2}} \right)} \right)} \right)
	\end{equation}
\end{proof}
Next we find $\beta$ so that $\omega_{g}>0$. 
The expression for $\omega_g$ is given in (\ref{wg_sol_final}). For better understanding, let us denote
$\Omega = \cos ^{ - 1}\left( {\frac{c_1}{r_1}} \right)$. Then using (\ref{c_1/r_1_final}), we have 
\begin{equation} \label{Omega}
\Omega = \cos ^{ - 1}\left( {\frac{c_1}{r_1}} \right) = \cos^{-1}\left({\frac{{\sin \left( {\frac{{\beta \pi }}{2}} \right)}}{{2\sin \left( {\frac{{{\phi _m}}}{2}} \right)}}}\right)
\end{equation}
Also denoting $\beta\pi/2 + \phi_{m}/2 = \eta$ in (\ref{alpha1_final}), we can write
\begin{equation} \label{alpha1_in_eta}
\alpha_1  = \left\{ {\begin{array}{*{20}{c}}
	{ - \eta }&{\eta  \in \textcolor{black}{\left( {0,\frac{\pi }{2}} \right)}}\\
	{\pi  - \eta }&{\eta  \in \left( {\frac{\pi }{2},\frac{{3\pi }}{2}} \right)}
	\end{array}} \right.
\end{equation}
With the new notations, the expression of $\omega_{g}$ in (\ref{wg_sol_final}) becomes,
\begin{equation} \label{wg_in_Omeag_and_alpha1}
{\omega _g} = \frac{1}{\theta }\left( { - \Omega  + {\alpha _1}} \right)
\end{equation}
In accordance with the controller design, $\phi_m\in(0,\pi)$ and $\beta\in(0,2)$.
From (\ref{c_1/r_1_final}), with arbitrary $\phi_m\in(0,\pi)$, we have $c_1/r_1 \in (0,1/\left( {2\sin \left( {{\phi _m}/2} \right)} \right) ]$ for both  $\beta \in(0,1]$ and $\beta \in(1,2)$.
Hence, it can be noted that $c_1/r_1  > 0 \, \forall \, \beta \in(0,2)$.
On the other hand,  
\textcolor{black}{to have real solution for $\omega_{g}$, we need }
$\left| c_1/r_1 \right| \le 1$. Therefore the common admissible range of $c_1/r_1$ is $(0,1]$.
Therefore, $\Omega \in(0,\pi/2]$ from (\ref{Omega}).

From (\ref{wg_in_Omeag_and_alpha1}), for $\omega_g>0$ we need $\Omega<\alpha_1$. However, it is already shown in previous paragraph that $\Omega\in(0,\pi/2]$. Whereas from (\ref{alpha1_in_eta}), $\eta \in(0,3\pi/2)$.

To find ${{\rm B} _{{\omega _g} + }} = \left\{ {\beta :{\omega _g} > 0} \right\}$, following three cases need to be considered:
\begin{enumerate}[(i)]
	\item    $\Omega\in(0,\pi/2]$ and $\eta \in \textcolor{black}{(0,\pi/2)}$ $\Rightarrow {\alpha _1} \in \left( { - \pi /2,0} \right)$
	\item 	 $\Omega\in(0,\pi/2]$ and $\eta \in (\pi/2,\pi]$ $ \Rightarrow {\alpha _1} \in \left[ {0,\pi /2} \right) $
	\item 	 $\Omega\in(0,\pi/2]$ and $\eta \in (\pi,3\pi/2)$ $ \Rightarrow {\alpha _1} \in \left( { - \pi /2,0} \right) $
\end{enumerate} 
\textcolor{black}{It is clear that cases $(i)$ and $(iii)$ will lead to negative $\omega_{g}$ as $\Omega>\alpha_1$ in the entire range of $\Omega$ and $\alpha_1$. Hence they are not considered in further analysis.}
\begin{lemma}
	[Existence of solution in $case-ii$] The solution set $\beta_{{\omega _g} + }$ exists only when $\Omega\in(0,\pi/2]$ and $\eta\in(\pi/2,\pi]$. 
\end{lemma}
\begin{proof}
	\textcolor{black}{From (\ref{alpha1_in_eta}), for \textcolor{black}{$\eta \in (\pi/2,\pi]$}, we get $\alpha_1 \in [0,\pi/2)$.
		Therefore, for $\omega_g>0$, }
	\begin{equation} \label{wg_more_zero_case_2_initial}
	{\cos ^{ - 1}}\left( {\frac{{\sin \left( {\frac{{\beta \pi }}{2}} \right)}}{{2\sin \left( {\frac{{{\phi _m}}}{2}} \right)}}} \right) < \pi  - \eta
	\end{equation}
	Above equation  
	can be solved by taking $\cos$ in both side while considering LHS and RHS as arguments.
	In (\ref{wg_more_zero_case_2_initial}), LHS and RHS both are in $(0,\pi/2)$ and in this range $\cos$ is decreasing function. So taking $\cos$ both side, the inequality is reversed and we get 
	\begin{equation} \label{wg_more_zero_case_2_2}
	\frac{{\sin \left( {\frac{{\beta \pi }}{2}} \right)}}{{2\sin \left( {\frac{{{\phi _m}}}{2}} \right)}} > \cos \left( {\pi  - \eta} \right)
	\end{equation} 
	Since, $\cos(\pi-x)=-\cos(x)\,;\,x\in(\pi/2,\pi)$, therefore, simplifying (\ref{wg_more_zero_case_2_2}) \textcolor{black}{and replacing $\eta = {\frac{{\beta \pi }}{2} + \frac{{{\phi _m}}}{2}}$,} we get
	\begin{equation} \label{wg_more_zero_case_2_3}
	\frac{{\sin \left( {\frac{{\beta \pi }}{2}} \right)}}{{2\sin \left( {\frac{{{\phi _m}}}{2}} \right)}} > -\cos \left( {\frac{{\beta \pi }}{2} + \frac{{{\phi _m}}}{2}} \right)
	\end{equation} 
	Expanding the RHS and multiplying both side of the inequality by ${2\sin \left( {\frac{{{\phi _m}}}{2}} \right)}$, we have
	\begin{equation} \label{wg_more_zero_case_2_6}
	\sin \left( {\frac{{\beta \pi }}{2}} \right) > - \sin \left( {{\phi _m}} \right)\cos \left( {\frac{{\beta \pi }}{2}} \right) + 
	2{\sin ^2}\left( {\frac{{{\phi _m}}}{2}} \right)\sin \left( {\frac{{\beta \pi }}{2}} \right)
	\end{equation}
	Further, (\ref{wg_more_zero_case_2_6}) can be simplified by dividing by $\cos(\beta\pi/2)$. However, for $\beta\in(0,1]$, $\cos(\beta\pi/2)\ge 0$. Therefore the inequality sign will not be affected. However, for $\beta \in (1,2)$, $\cos(\beta\pi/2) < 0$, and the inequality sign will be reversed. So further classification can be done based on the \textcolor{black}{range} of $\beta$.
	
	\textit{Case-ii(a): $\beta \in(0,1]$ :}\\
	Dividing both side of (\ref{wg_more_zero_case_2_6}) by $\cos(\beta\pi/2)$ and simplifying, we get
	\begin{equation} \label{wg_more_zero_case_2(a)_1a}
	\tan \left( {\frac{{\beta \pi }}{2}} \right)\left( {1 - 2{{\sin }^2}\left( {\frac{{{\phi _m}}}{2}} \right)} \right) >  - \sin \left( {{\phi _m}} \right)
	\end{equation}
	The plot for $1 - 2{\sin ^2}\left( {{\phi _m}/2} \right)$ is given in Fig.\ref{fig_pm_vs_betay1_betax1_etc}(a), which is positive when $\phi_{m}\in(0,\pi/2]$ and negative when $\phi_{m}\in(\pi/2,\pi)$. Therefore, $case-ii(a)$ can be further divided as follows.
	
	\textit{$Case-ii(a_1)$, ($\beta\in(0,1]$ and $\phi_{m}\in(0,\pi/2]$):}  In this case, $1 - 2{\sin ^2}\left( {{\phi _m}/2} \right)$ is positive. Dividing both side of (\ref{wg_more_zero_case_2(a)_1a}) by $1 - 2{\sin ^2}\left( {{\phi _m}/2} \right)$, the inequality sign will remain same. Therefore, we have 
	\begin{equation} \label{wg_more_zero_case_2(a)_1}
	\tan \left( {\frac{{\beta \pi }}{2}} \right) > \frac{{ - \sin \left( {{\phi _m}} \right)}}{{1 - 2{{\sin }^2}\left( {\frac{{{\phi _m}}}{2}} \right)}}
	\end{equation}
	In (\ref{wg_more_zero_case_2(a)_1}), 
	LHS and RHS \textcolor{black}{are both} functions \textcolor{black}{with range} ($-\infty$,$\infty$). From fundamentals of inverse trigonometry we know that when $y=\arctan(x); x\in(-\infty,\infty)$ then $y\in[-\pi/2,\pi/2]$. Since $\arctan$ is increasing function in the fundamental period, therefore, the inequality sign in (\ref{wg_more_zero_case_2(a)_1}) will remain same on taking $\arctan$ on both side. Therefore, taking $\arctan$ in (\ref{wg_more_zero_case_2(a)_1}) and evaluating for $\beta$, we get
	\begin{equation} \label{wg_more_zero_case_2(a)_2}
	\beta  > \frac{2}{\pi }\left( {{{\tan }^{ - 1}}\left( {\frac{{ - \sin \left( {{\phi _m}} \right)}}{{1 - 2{{\sin }^2}\left( {\frac{{{\phi _m}}}{2}} \right)}}} \right)} \right) = {\beta _{y1}}(say)
	\end{equation}
	Then, we get the solution set 
	\textcolor{black}{${{\rm{B}}_{y11}} = \left\{ {\beta  \in \left( {0,1} \right]:\beta  > {\beta _{y1}},{\phi _m} \in \left( {0,\pi /2} \right],{\omega _g} > 0} \right\}$}.\\
	
	\textit{Case-ii$(a_2)$, ($\beta\in(0,1)$ and $\phi_{m}\in(\pi/2,\pi)$):}
	In this case $1 - 2{\sin ^2}\left( {{\phi _m}/2} \right)$ is negative. Therefore, dividing both side by $1 - 2{\sin ^2}\left( {{\phi _m}/2} \right)$ in  (\ref{wg_more_zero_case_2(a)_1a}), the inequality sign will get reverse and we have,
	\begin{equation} \label{wg_more_zero_case_2(a)_3}
	\tan \left( {\frac{{\beta \pi }}{2}} \right) < \frac{{ - \sin \left( {{\phi _m}} \right)}}{{1 - 2{{\sin }^2}\left( {\frac{{{\phi _m}}}{2}} \right)}}
	\end{equation}
	In the similar fashion as in $case-ii(a_1)$, LHS and RHS in (\ref{wg_more_zero_case_2(a)_3}) are functions with \textcolor{black}{range} in $(-\infty,\infty)$. Taking $\arctan$ on both side does not affect the inequality sign in this case. Evaluating for $\beta$, we have
	\begin{equation} \label{wg_more_zero_case_2(a)_4}
	\beta  < \frac{2}{\pi }\left( {{{\tan }^{ - 1}}\left( {\frac{{ - \sin \left( {{\phi _m}} \right)}}{{1 - 2{{\sin }^2}\left( {\frac{{{\phi _m}}}{2}} \right)}}} \right)} \right) = {\beta _{y1}}
	\end{equation}
	Therefore, we get the solution set 
	$\textcolor{black}{{{\rm{B}}_{y12}} = \left\{ {\beta  \in \left( {0,1} \right]:\beta  < {\beta _{y1}},{\phi _m} \in \left( {\pi /2,\pi } \right),{\omega _g} > 0} \right\}}$.\\
	
	\textit{Case-ii(b): $\beta \in(1,2)$ :}\\
	In (\ref{wg_more_zero_case_2_6}), dividing both side by $\cos(\beta\pi/2)$, we get
	\begin{equation}  \label{wg_more_zero_case_2(b)_1}
	\tan \left( {\frac{{\beta \pi }}{2}} \right) <  - \sin \left( {{\phi _m}} \right) + 2\tan \left( {\frac{{\beta \pi }}{2}} \right){\sin ^2}\left( {\frac{{{\phi _m}}}{2}} \right)
	\end{equation}
	Simplifying (\ref{wg_more_zero_case_2(b)_1}), we get
	\begin{equation} \label{wg_more_zero_case_2(b)_1a}
	\tan \left( {\frac{{\beta \pi }}{2}} \right)\left( {1 - 2{{\sin }^2}\left( {\frac{{{\phi _m}}}{2}} \right)} \right) <  - \sin \left( {{\phi _m}} \right)
	\end{equation}
	Again, from Fig.\ref{fig_pm_vs_betay1_betax1_etc}(a), $1 - 2{\sin ^2}\left( {{\phi _m}/2} \right)$ is negative for $\phi_{m}\in(0,\pi/2]$ and positive for $\phi_{m}\in(\pi/2,\pi)$. Therefore, $case-ii(b)$ can be further divided into two sub cases according to the range of $\phi_{m}$.\\
	
	\textit{case-ii$(b_1)$, ($\beta$$\in$(1,2) and $\phi_{m}\in(0,\pi/2]$):}
	In this case, $1 - 2{\sin ^2}\left( {{\phi _m}/2} \right)$ is positive, therefore, dividing both side in (\ref{wg_more_zero_case_2(b)_1a}) by $1 - 2{\sin ^2}\left( {{\phi _m}/2} \right)$, the inequality sign will not be affected. Therefore, we get
	\begin{equation} \label{wg_more_zero_case_2(b)_2}
	\tan \left( {\frac{{\beta \pi }}{2}} \right) < \frac{{ - \sin \left( {{\phi _m}} \right)}}{{1 - 2{{\sin }^2}\left( {\frac{{{\phi _m}}}{2}} \right)}}
	\end{equation}
	To evaluate $\beta$ from (\ref{wg_more_zero_case_2(b)_2}), we need to take $\arctan$ on both side. From fundamental of trigonometry, ${\tan ^{ - 1}}\left( {\tan \left( x \right)} \right) = x | x\in(-\pi/2,\pi/2)$. Notice that the argument  $\beta\pi/2$ lies in $(\pi/2,\pi)$ for $\beta \in (1,2)$, which is outside the range of fundamental period.
	In this case the origin needs to be shifted as $\tan \left( {\frac{{\beta \pi }}{2}} \right) = \tan \left( {\frac{{\beta \pi }}{2} - \pi } \right)$. Therefore, \textcolor{black}{to get correct solution} from (\ref{wg_more_zero_case_2(b)_2}), we need to solve
	\begin{equation} \label{wg_more_zero_case_2(b)_3}
	\tan \left( {\frac{{\beta \pi }}{2} - \pi } \right) < \frac{{ - \sin \left( {{\phi _m}} \right)}}{{1 - 2{{\sin }^2}\left( {\frac{{{\phi _m}}}{2}} \right)}}
	\end{equation}

	In (\ref{wg_more_zero_case_2(b)_3}), for $\beta \in (1,2)$, $\tan(\beta\pi/2-\pi) \in (-\infty,0)$ and $\arctan(-\infty,0) \in (0,\pi/2)$. However, in $(0,\pi/2)$,  $\arctan$ is increasing function. Therefore, the inequality sign will remain same on taking $\arctan$ on both side.
	
	Therefore, taking $\arctan$ on both side of (\ref{wg_more_zero_case_2(b)_3}), we get
	\[\frac{{\beta \pi }}{2} - \pi  < {\tan ^{ - 1}}\left( {\frac{{ - \sin \left( {{\phi _m}} \right)}}{{1 - 2{{\sin }^2}\left( {\frac{{{\phi _m}}}{2}} \right)}}} \right)\]
	Evaluating for $\beta$, we get
	\begin{equation} \label{wg_more_zero_case_2(b)_4}
	\beta  < \frac{2}{\pi }\left( {\pi  + {{\tan }^{ - 1}}\left( {\frac{{ - \sin \left( {{\phi _m}} \right)}}{{1 - 2{{\sin }^2}\left( {\frac{{{\phi _m}}}{2}} \right)}}} \right)} \right) = \beta_{y2}(say)
	\end{equation}
	Then we have the solution set as
	$\textcolor{black}{{{\rm{B}}_{y21}} = \left\{ {\beta  \in \left( {1,2} \right):\beta  < {\beta _{y2}},{\phi _m} \in \left( {0,\pi /2} \right],{\omega _g} > 0} \right\}}$.\\
	
	\textit{case-ii$(b_2)$, ($\beta\in(1,2)$ and $\phi_{m}\in(\pi/2,\pi)$):}
	In this case,  $1 - 2{\sin ^2}\left( {{\phi _m}/2} \right)$ is negative. Therefore, dividing both side in (\ref{wg_more_zero_case_2(b)_1a}) by $1 - 2{\sin ^2}\left( {{\phi _m}/2} \right)$ the sign of inequality will get reversed. Hence, we have
	\begin{equation} \label{wg_more_zero_case_2(b)_5}
	\tan \left( {\frac{{\beta \pi }}{2}} \right) > \frac{{ - \sin \left( {{\phi _m}} \right)}}{{1 - 2{{\sin }^2}\left( {\frac{{{\phi _m}}}{2}} \right)}}
	\end{equation}
	Again, for $\beta\in(1,2)$, $\beta\pi/2 \in (\pi/2,\pi)$ which is outside the range of fundamental period of $\tan(x)$. Therefore, \textcolor{black}{we shift the origin and get}
	\begin{equation} \label{wg_more_zero_case_2(b)_6}
	\tan \left( {\frac{{\beta \pi }}{2} - \pi } \right) > \frac{{ - \sin \left( {{\phi _m}} \right)}}{{1 - 2{{\sin }^2}\left( {\frac{{{\phi _m}}}{2}} \right)}}
	\end{equation}
	In this range $\arctan$ is increasing function. Thus, we get,
	\begin{equation} \label{wg_more_zero_case_2(b)_7}
	\beta  > \frac{2}{\pi }\left( {\pi  + {{\tan }^{ - 1}}\left( {\frac{{ - \sin \left( {{\phi _m}} \right)}}{{1 - 2{{\sin }^2}\left( {\frac{{{\phi _m}}}{2}} \right)}}} \right)} \right) = \beta_{y2}
	\end{equation}
	Therefore, the solution set is
	$\textcolor{black}{{{\rm{B}}_{y22}} = \left\{ {\beta  \in \left( {1,2} \right):\beta  > {\beta _{y2}},{\phi _m} \in \left( {\pi /2,\pi } \right),{\omega _g} > 0} \right\}}$.\\
	\begin{figure}  
		\centering
		\begin{tabular}{ c l }
			\subfigure[]{
				\includegraphics[width=1.7in] {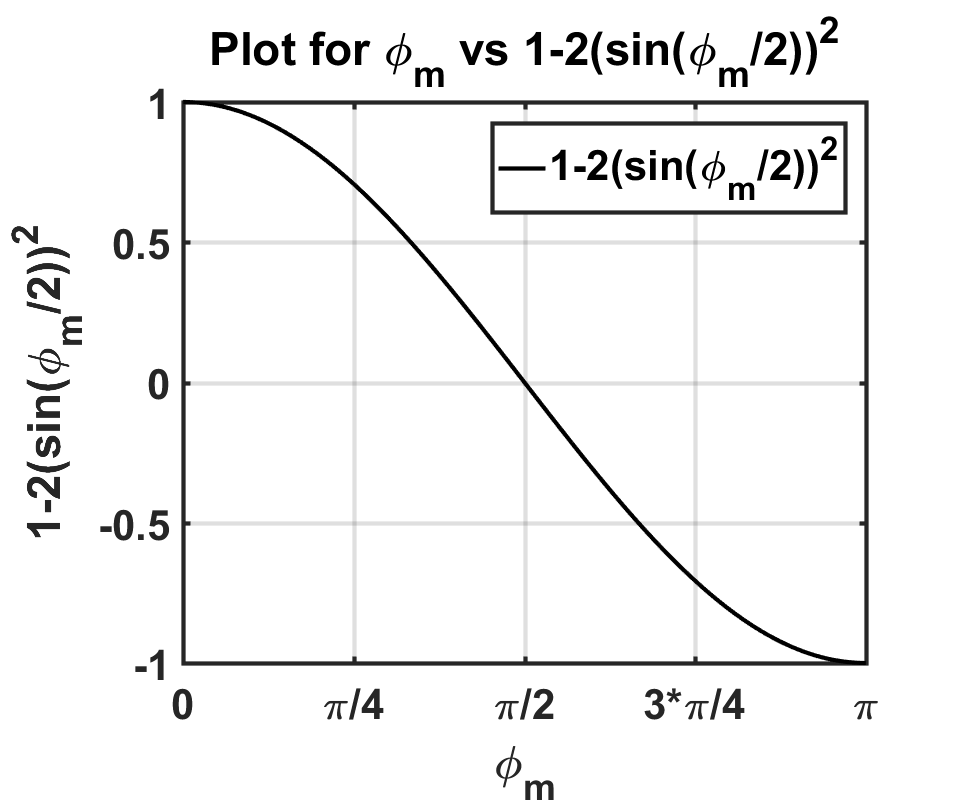}}
			
			\subfigure[]{
				\includegraphics[width=1.7in] {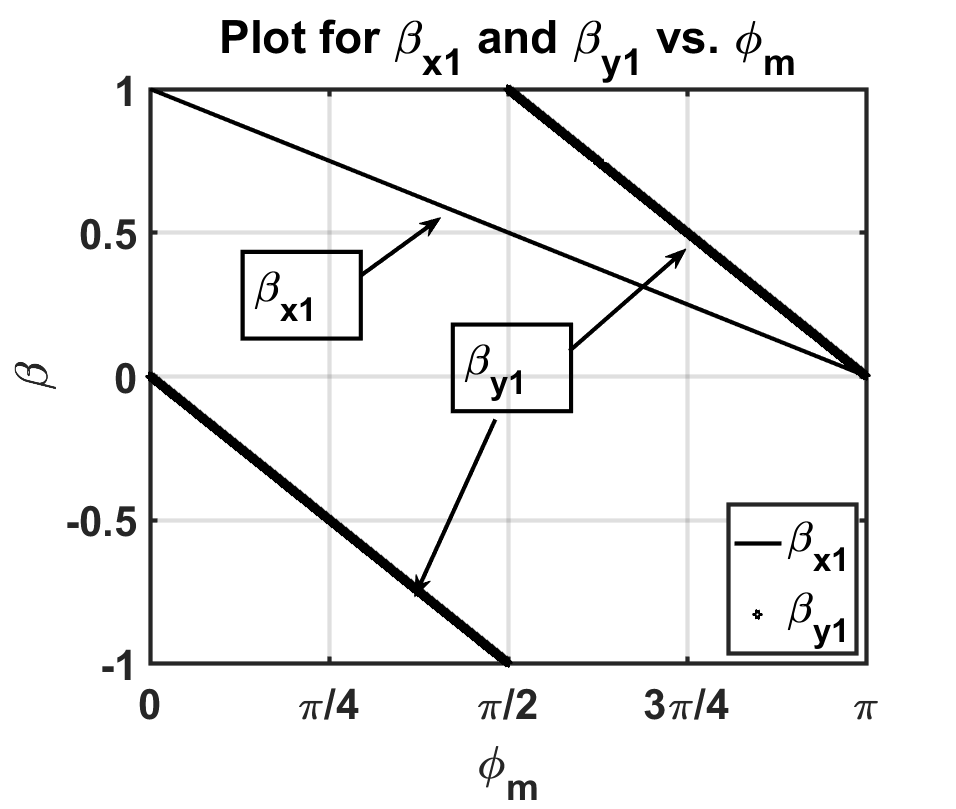}} 
			
			\subfigure[]{
				\includegraphics[width=1.8in] {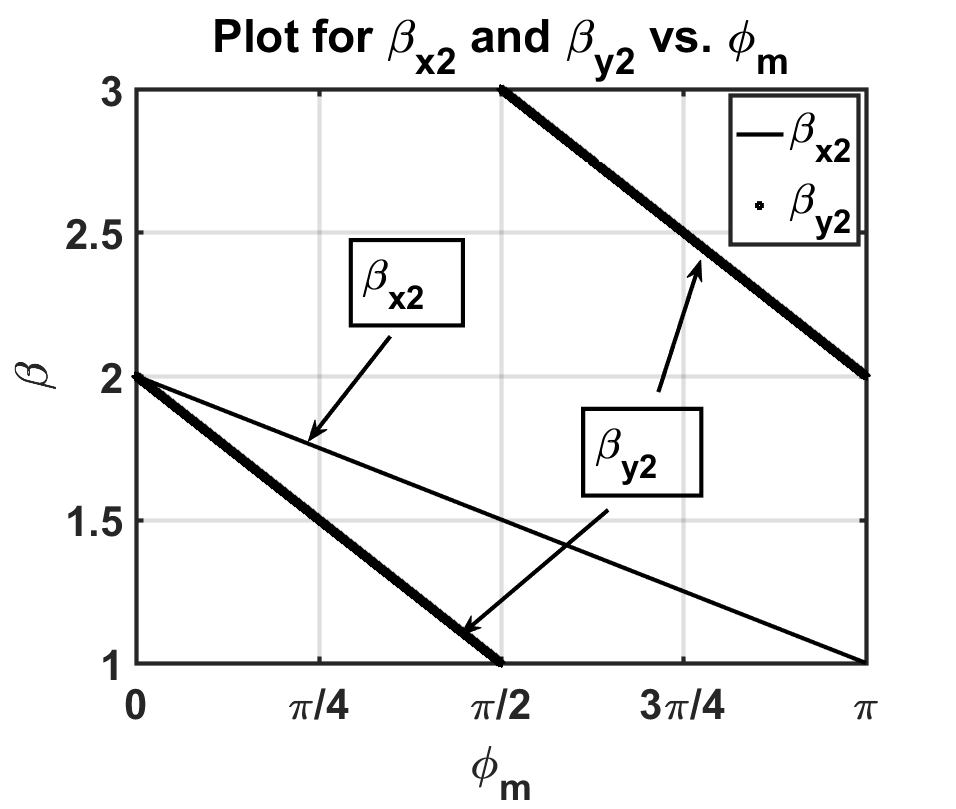}}
			
		\end{tabular}
		\caption{Graphical representation of (a) $\phi_m$ vs. $1 - 2{\sin ^2}\left( {{\phi _m}/2} \right)$, (b) $\phi_m$ vs. $\beta_{x1}$ and $\phi_{m}$ vs. $\beta_{y1}$, (c) $\phi_m$ vs. $\beta_{x2}$ and $\phi_{m}$ vs. $\beta_{y2}$.}
		\label{fig_pm_vs_betay1_betax1_etc}
	\end{figure}
	The above four cases are summarized in  Table-\ref{table-1}. 
	\textcolor{black}{From Fig.\ref{fig_pm_vs_betay1_betax1_etc}(c), we can see that no $\beta\in(1,2)$ for $\phi_{m}\in(\pi/2,\pi)$ exists such that $\beta>\beta_{y2}$. Hence, there is no solution for $case-ii(b_2)$. From the same figure, it is evident that the other three cases can have solution for $\beta$.}
	\begin{table}[]
		\centering
		\caption{Existence of solution such that $\omega_{g}>0$}
		\label{table-1}
		\begin{tabular}{|l|l|l|}
			\hline 
			& ${{\phi _m} \in \left( {0,\pi /2} \right]}$      & ${{\phi _m} \in \left( {\pi /2,\pi } \right)}$      \\ \hline
			${\beta  \in \left( {0,1} \right]}$    &  $case-ii(a_1)$, ${\rm B}_{y11}={\beta  > {\beta _{y1}}}$    &$case-ii(a_2)$, $ {\rm B}_{y12}={\beta  < {\beta _{y1}}}$      \\ \hline
			${\beta  \in \left( {1,2} \right)}$      & $case-ii(b_1)$, ${\rm B}_{y21}={\beta  < {\beta _{y2}}}$    &$case-ii(b_2)$, ${\rm B}_{y22}={\beta  > {\beta _{y2}}} $      \\ \hline
		\end{tabular}
	\end{table}
\end{proof}
\textcolor{black}{Lemma-3 gives the possible solution for ${\rm B}_{{\omega _g} + }$ when $\eta\in(\pi/2,\pi]$. Since, $\eta = (\beta\pi/2 + \phi_{m}/2)$, this implies $\beta\in ( (\pi-\phi_{m})/\pi, (2\pi-\phi_{m})/\pi ]$. Let a set ${{\rm B}_x} = \left\{ {\beta :\beta  \in \left( {{\beta _{{x1}}},{\beta _{x2}}} \right]} \right\}$, where ${\beta _{{x1}}} = ((\pi-\phi_{m})/\pi )$ and ${\beta _{{x2}}} = (2\pi-\phi_{m})/\pi$. Therefore, ${\rm B}_{{\omega _g} + }$ can be found from intersection of solution set in Lemma-3 and the set ${{\rm B}_x}$. Therefore, the task that remains is to find ${\rm B}_{{\omega _g} + }$ for each of the valid three cases. }\\
\textcolor{black}{Note: For $\phi_{m}\in(0,\pi)$, $\beta_{x1} \in (0,1)$ and $\beta_{x2} \in (1,2)$.}
\begin{theorem} 	[Finding $ {\rm B}_{{\omega _g} + }$]
	\textcolor{black}{The set ${\rm B}_{{\omega _g} + }$ containing solution $\beta$ when $\omega_{g}>0$ is given as $\textcolor{black}{{{\rm B}_{{\omega _g} + }} = \left\{ {\beta :\beta  \in ({\beta _{x1}},{\beta _{y2}})} \right\}}$ when $\phi_{m}\in(0,\pi/2]$ and $\textcolor{black}{{{\rm B}_{{\omega _g} + }} = \left\{ {\beta :\beta  \in ({\beta _{x1}},{\beta _{y1}})} \right\}}$ \\
		when $\phi_{m}\in(\pi/2,\pi)$},
	where\\
	$\beta_{x_1} = (\pi-\phi_{m})/\pi$,
	${\beta _{y1}} = \frac{2}{\pi }\left( {{{\tan }^{ - 1}}\left( {\frac{{ - \sin \left( {{\phi _m}} \right)}}{{1 - 2{{\sin }^2}\left( {\frac{{{\phi _m}}}{2}} \right)}}} \right)} \right)$ and 
	$\beta_{y2} = \frac{2}{\pi }\left( {\pi  + {{\tan }^{ - 1}}\left( {\frac{{ - \sin \left( {{\phi _m}} \right)}}{{1 - 2{{\sin }^2}\left( {\frac{{{\phi _m}}}{2}} \right)}}} \right)} \right)$ 
	are given in (\ref{wg_sol_final_expanded}), (\ref{wg_more_zero_case_2(a)_2}) and (\ref{wg_more_zero_case_2(b)_4}) respectively.
\end{theorem}
\begin{proof}
	\textbf{\textit{$Case-ii(a_1)$, ($\beta\in(0,1]$, $\phi_{m}\in(0,\pi/2]$):}} 
	${\rm B}_{{\omega _g} + } = {\rm B}_x \cap {\rm B}_{y11}$ where, ${{\text{B}}_x} = \left\{ {\beta :\beta  \in \left( {{\beta _{x1}},{\beta _{x2}}} \right]} \right\}$ and  ${{\text{B}}_{y11}} = \left\{ {\beta :\beta  > {\beta _{y1}}} \right\}$. From Fig.\ref{fig_pm_vs_betay1_betax1_etc}(a), $\beta_{y1} \in (-1,0)$, whereas $\beta_{x1}\in(0,1)$ and $\beta_{x2} \in(1,2)$. Therefore, in this scenario we have
	$\textcolor{black}{ {{\rm B}_x} = \left\{ {\beta :\beta  \in ({\beta _{x1}},1]} \right\} }$ and \textcolor{black}{${{\rm B}_{y11}} = \left\{ {\beta :\beta  \in (0,1]} \right\}$}. Hence, the intersection of ${\rm B}_x$ and ${\rm B}_{y11}$ will be the solution set $\textcolor{black}{{{\rm B}_{{\omega _g} + }} = \left\{ {\beta :\beta  \in ({\beta _{x1}},1]:{\phi _m} \in (0,\pi /2]} \right\}}$.
	
	\textbf{\textit{$Case-ii(b_1)$, ($\beta\in(1,2), \phi_{m}\in(0,\pi/2]$):}}
	${\rm B}_{{\omega _g} + } = {\rm B}_x \cap {\rm B}_{y21}$, where, ${{\rm B}_{y21}} = \left\{ {\beta :\beta  < {\beta _{y2}}} \right\}$. From Fig.\ref{fig_pm_vs_betay1_betax1_etc}(b), for $\phi_{m}\in(0,\pi/2)$, $\beta_{y2} \in (1,2)$, whereas $\beta_{x2}\in(1,2)$. 
	Therefore, in this scenario we have \textcolor{black}{${{\rm B}_x} = \left\{ {\beta :\beta  \in (1,{\beta _{x2}})} \right\}$} and 
	\textcolor{black}{${{\rm B}_{y21}} = \left\{ {\beta :\beta  \in (1,{\beta _{y2}})} \right\}$}.
	Therefore,
	${{\rm B}_{{\omega _g} + }} = {\rm B}_x \cap {\rm B}_{y21}=  \left\{ {\beta :\beta  \in \left( {1,\min \left( {{\beta _{x2}},{\beta _{y2}}} \right)} \right)} \right\}$. 
	From, Fig.\ref{fig_pm_vs_betay1_betax1_etc}(b), for $\phi_{m}\in(0,\pi/2)$, $\beta_{y2}<\beta_{x2}$. Therefore, ${{\rm B}_{{\omega _g} + }} = \left\{ {\beta :\beta  \in \left( {1,{\beta _{y2}}} \right)} \right\}$.
	
	Combining $Case-ii(a_1)$ and $Case-ii(b_1)$,
	for $\phi_{m}\in(0,\pi/2]$,
	${{\rm B}_{{\omega _g} + }} = \left\{ {\beta :\beta  \in \left( {{\beta _{x1}},1} \right] \cup \left( {1,{\beta _{y2}}} \right)} \right\} = \{ \beta :  \beta \in \left( {{\beta _{x1}},{\beta _{y2}}} \right) \}$.
	
	\textbf{\textit{$Case-ii(a_2)$, ($\beta\in(0,1), \phi_{m}\in(\pi/2,\pi)$):}}
	${\rm B}_{{\omega _g} + } = {\rm B}_x \cap {\rm B}_{y12}$, where, ${{\rm B}_{y12}} = \left\{ {\beta :\beta  < {\beta _{y1}}} \right\}$. From Fig.\ref{fig_pm_vs_betay1_betax1_etc}(a), for $\phi_{m}\in(0,\pi/2)$, $\beta_{y1} \in (0,1)$, whereas $\beta_{x1}\in(0,1)$ and $\beta_{x2}\in(1,2)$.
	Therefore, in this scenario we have \textcolor{black}{${{\rm B}_x} = \left\{ {\beta :\beta  \in ({\beta _{x1}}, 1]} \right\}$} and \textcolor{black}{${{\text{B}}_{y12}} = \left\{ {\beta :\beta  \in \left( {0,{\beta _{y1}}} \right)} \right\}$}.
	For existence of solution, we must have $\beta_{x1}<\beta_{y1}$. From Fig.\ref{fig_pm_vs_betay1_betax1_etc}(a), for $\phi_{m}\in(\pi/2,\pi)$, $\beta_{x1}<\beta_{y1}$. Therefore, ${{\rm B}_{{\omega _g} + }} = \left\{ {\beta :\beta  \in \left( {{\beta _{x1}},{\beta _{y1}}} \right)} \right\}$ for $\phi_{m} \in (\pi/2 , \pi)$.  
\end{proof}
Now $\beta_{{\omega _g} }$ can be found by combining Theorem-2 and 3.
\begin{theorem}[Finding ${\rm B}_{{\omega _g} }$]
	The set ${\rm B}_{{\omega _g} }$ such that $\omega_{g}\in \Re^{+}$ is given as:\\
	(a) ${{\rm B}_{{\omega _g}}} = \left\{ {\beta :\beta  \in \left( {{\beta _{{\omega _g}\Re 2}},{\beta _{y2}}} \right)} \right\}$ when $\phi_{m}\in(0,0.9273]$,\\
	(b) ${{\rm B}_{{\omega _g}}} = \left\{ {\beta :\beta  \in \left( {{\beta _{x1}},{\beta _{{\omega _g}\Re 1}}} \right) \cup \left( {{\beta _{{\omega _g}\Re 2}},{\beta _{y2}}} \right)} \right\}$ when $\phi_{m}\in (0.9273,\pi/3)$,\\
	(c) ${{\rm B}_{{\omega _g}}} = \left\{ {\beta :\beta  \in \left( {{\beta _{x1}},{\beta _{y2}}} \right)} \right\}$ when $\phi_{m} \in [\pi/3, \pi/2]$, and\\
	(d) ${{\rm B}_{{\omega _g}}} = \left\{ {\beta :\beta  \in \left( {{\beta _{x1}},{\beta _{y1}}} \right)} \right\}$ when $\phi_{m}\in(\pi/2, \pi)$.
\end{theorem}
\begin{proof}
	From Theorem-2,
	$ {{\rm{B}}_{{\omega _g}\Re }} = \{ \beta :\beta  \in (0,{\beta _{{\omega _g}\Re 1}}) \cup ({\beta _{{\omega _g}\Re 2}},2)\} $
	for 
	$\phi_{m}\in(0,\pi/3)$, whereas,
	${{\rm{B}}_{{\omega _g}\Re }} = \{ \beta :\beta  \in (0,2)\} $
	for $\phi_{m}\in[\pi/3,\pi)$.
	However, from Theorem-3,
	${{\rm{B}}_{{\omega _g} + }} = \{ \beta :\beta  \in ({\beta _{x1}},{\beta _{y2}})\} $
	if $\phi_{m}\in(0,\pi/2]$,
	whereas
	${{\rm{B}}_{{\omega _g} + }} = \{ \beta :\beta  \in ({\beta _{x1}},{\beta _{y1}})\} $
	if $\phi_{m}\in (\pi/2,\pi)$. 
	Therefore, combining Theorem 2 and 3, following three cases are possible.
	
	\textit{Case-I, $\phi_{m}\in(0,\pi/3)$:} In this case,
	${{\rm{B}}_{{\omega _g}\Re }} = \{ \beta :\beta  \in (0,{\beta _{{\omega _g}\Re 1}}) \cup ({\beta _{{\omega _g}\Re 2}},2)\} $
	and
	$ {{\rm{B}}_{{\omega _g} + }} = \{ \beta :\beta  \in ({\beta _{x1}},{\beta _{y1}})\} $.
	For existence of solution in $\beta\in(0,1)$, we must have $\beta_{x1} < \beta_{{\omega _g}\Re 1}$ and for existence of solution in $\beta(1,2)$, $\beta_{{\omega _g}\Re 2} < \beta_{y2}$ is needed.
	
	In Fig.\ref{pm_vs_betay1_betawgR1_etc}(a), $\beta_{x1} - \beta_{{\omega _g}\Re 1}$ plot and in Fig.\ref{pm_vs_betay1_betawgR1_etc}(b), $\beta_{y2} - \beta_{{\omega _g}\Re 2}$ plot with respect to $\phi_{m}\in(0,\pi/3)$ is given. 	
	It can be seen that  $\beta_{x1} < \beta_{{\omega _g}\Re 1}$ is satisfied for $\phi_{m}\in(0.9273,\pi/3)$, and if $\phi_{m}\in(0,0.9273]$, there is no $\beta \in (0,1)$ such that $\omega_{g}>0$.  
	However, from Fig.\ref{pm_vs_betay1_betawgR1_etc}(b), $\beta_{{\omega _g}\Re 2} < \beta_{y2} \forall \phi_{m}\in(0,\pi/3)$. Therefore,
	${{\rm{B}}_{{\omega _g}}} = \{ \beta :\beta  \in ({\beta _{{\omega _g}\Re 2}},{\beta _{y2}})\} $ in this case.
	
	This concludes that if $\phi_{m}\in(0,0.9273]$, then ${{\rm B}_{{\omega _g}}} = \left\{ {\beta :\beta  \in \left( {{\beta _{{\omega _g}\Re 2}},{\beta _{y2}}} \right)} \right\}$ and if $\phi_{m}\in (0.9273,\pi/3)$, then ${{\rm B}_{{\omega _g}}} = \left\{ {\beta :\beta  \in \left( {{\beta _{x1}},{\beta _{{\omega _g}\Re 1}}} \right) \cup \left( {{\beta _{{\omega _g}\Re 2}},{\beta _{y2}}} \right)} \right\}$. 
	
	\textcolor{black}{Part} (a) and (b) of the theorem \textcolor{black}{is hence proved}.
	
	\begin{figure}  [h]
		\centering
		\begin{tabular}{ c c }
			\subfigure[]{
				\includegraphics[width=1.6in] {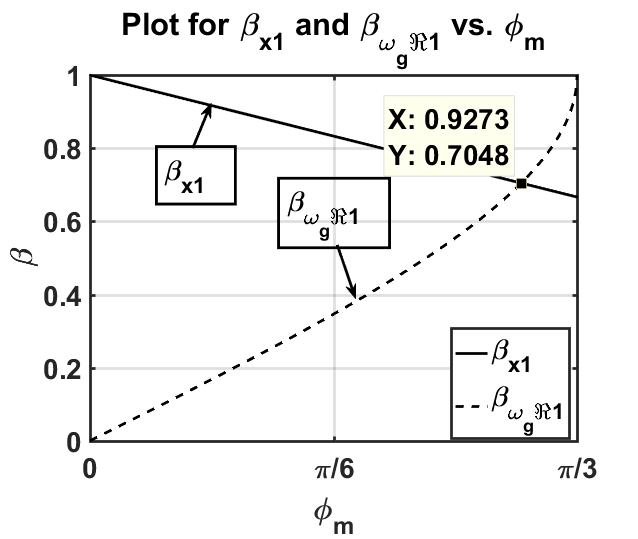}}
			&
			\subfigure[]{
				\includegraphics[width=1.6in] {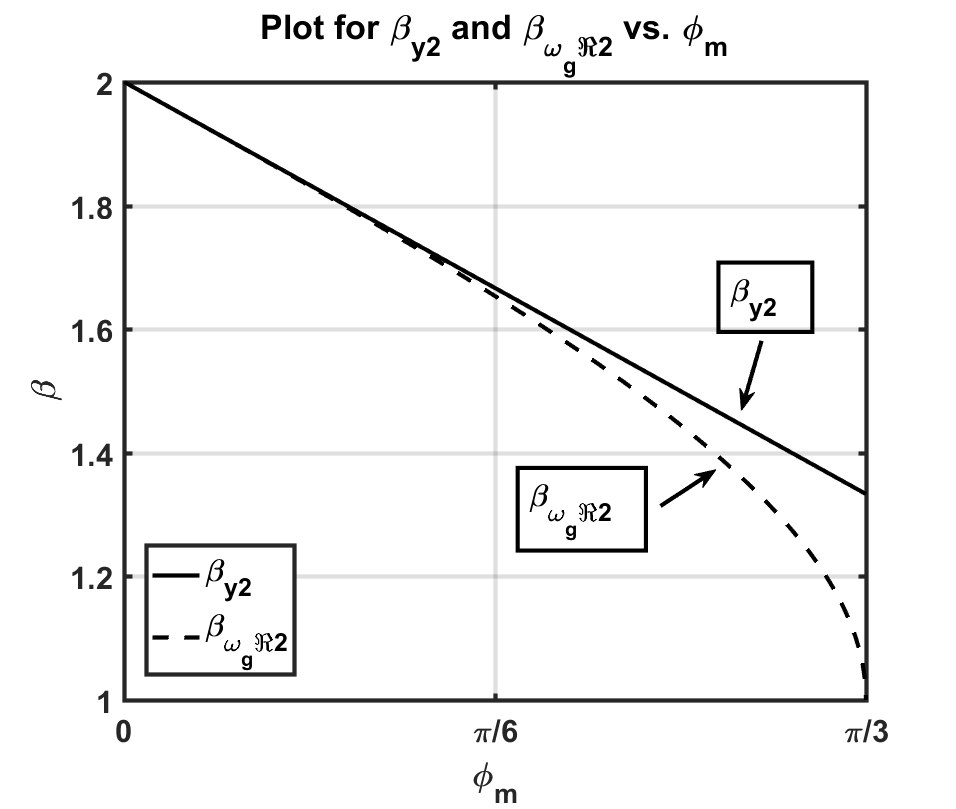}}
		\end{tabular}
		\caption{Graphical representation of (a) $\phi_m$ vs. $\beta_{x1}$ and $\phi_{m}$ vs. $\beta_{{\omega _g}\Re 1}$ (b)  $\phi_{m}$ vs. $\beta_{y2}$ and $\phi_{m}$ vs. $\beta_{{\omega _g}\Re 2}$}
		\label{pm_vs_betay1_betawgR1_etc}
	\end{figure}
	
	\textit{Case-II, $\phi_{m}\in[\pi/3,\pi/2]$:} 
	In this case, ${{\rm B}_{{\omega _g}\Re }} = \left\{ {\beta :\beta  \in \left( {0,2} \right)} \right\}$ and ${\beta _{{\omega _g} + }} = \left\{ {\beta :\beta  \in ({\beta _{x1}},{\beta _{y2}})} \right\}$. Since, ${{\rm B}_{{\omega _g}}} = {{\rm B}_{{\omega _g}\Re }} \cap {{\rm B}_{{\omega _g} + }}$, therefore, 
	${{\rm{B}}_{{\omega _g}}} = {{\rm{B}}_{{\omega _g} + }} = \left\{ {\beta :\beta  \in \left( {{\beta _{x1}},{\beta _{y2}}} \right)} \right\}$. This proves part (c) of the theorem.
	
	\textit{Case-III, $\phi_{m}\in(\pi/2, \pi)$:} In this case, ${{\rm B}_{{\omega _g}\Re }} = \left\{ {\beta :\beta  \in \left( {0,2} \right)} \right\}$ and ${{\rm{B}}_{{\omega _{g + }}}} = \left\{ {\beta :\beta  \in \left( {{\beta _{x1}},{\beta _{y1}}} \right)} \right\}$. Therefore, ${{\rm{B}}_{{\omega _g}}} = {{\rm{B}}_{{\omega _g} + }} = \left\{ {\beta :\beta  \in \left( {{\beta _{x1}},{\beta _{y1}}} \right)} \right\}$. This proves part (d) of the theorem. 
\end{proof}
\subsection{Finding $\omega_{p}\in \Xi _b $:}
Eliminating $\lambda$ in (\ref{GM_condition_eq3_tosolve}) and (\ref{GM_condition_eq4_tosolve}), we get
\begin{equation} \label{wp_from_GM_condn}
\tan \left( {\frac{{\beta \pi }}{2}} \right) = \frac{{{A_m}\sin \left( {\theta {\omega _p}} \right) - \sin \left( {\theta {\omega _p}} \right)}}{{\cos \left( {\theta {\omega _p}} \right) - {A_m}\cos \left( {\theta {\omega _p}} \right) - 1}}
\end{equation}
In (\ref{wp_from_GM_condn}), there are four variables $\omega_g$, $\beta, \theta$  and $A_m$. Here $\theta$ is the delay of the process and known, $A_m$ is the desired gain margin specification and $\beta$ and $\omega_{p}$ are unknown variables that need to be evaluated.
Simplifying (\ref{wp_from_GM_condn}), we have
\begin{align*}
\begin{array}{c}
\tan \left( {\frac{{\beta \pi }}{2}} \right)\left( {1 - A_m} \right)\cos \left( {\theta {\omega _p}} \right) - \left( {A_m - 1} \right)\sin \left( {\theta {\omega _p}} \right) = 
\tan \left( {\frac{{\beta \pi }}{2}} \right)
\end{array}
\end{align*} 
\begin{align} \label{wp_equation_initial}
\text{or,} \qquad	\tan \left( {\frac{{\beta \pi }}{2}} \right)\cos \left( {\theta {\omega _p}} \right) + \sin \left( {\theta {\omega _p}} \right) = \frac{{\tan \left( {\frac{{\beta \pi }}{2}} \right)}}{{1 - {A_m}}}
\end{align}
which can be further simplified as
\begin{align}  \label{wp_in_a2b2c2}
{a_2}\cos \left( {\theta {\omega _p}} \right) + {b_2}\sin \left( {\theta {\omega _p}} \right) = {c_2} \\
\label{a2b2c2_values}
\text{where} \qquad
\begin{array}{c}
a_2 = \tan \left( {\frac{{\beta \pi }}{2}} \right); \,
b_2 = 1; \, 
c_2 = \frac{{\tan \left( {\frac{{\beta \pi }}{2}} \right)}}{{1 - {A_m}}}
\end{array}
\end{align}
Now, let $a_2 = r_2 \cos \left( \alpha_2  \right)$ and $b_2 = r_2 \sin \left( \alpha_2  \right)$, 
where $ {r_2} = \sqrt {a_2^2 + b_2^2}$ and ${\alpha _2} = {\tan ^{ - 1}}\left( {\frac{{{b_2}}}{{{a_2}}}} \right)$. 
Using trigonometric identity,
$\cos(x)\cos(y)+\sin(x)\sin(y) = \cos(x-y)$. Assuming $\theta\omega_p = x$ and $\alpha_2=y$,
(\ref{wp_in_a2b2c2}) becomes
\begin{equation} \label{wp_equation_initial_simplified}
\qquad \cos \left( {\theta {\omega _p} - {\alpha _2}} \right) = \frac{c_2}{r_2}
\end{equation} 
\begin{lemma}[Simplification for $c_2/r_2$]
	If $a_2$, $b_2$ and $c_2$ are as given in (\ref{a2b2c2_values}), then $c_2/r_2$ can be given as
	\begin{equation} \label{c2_by_r2_final}
	\frac{{{c_2}}}{{{r_2}}} =  - \frac{{\sin \left( {\frac{{\beta \pi }}{2}} \right)}}{{{A_m} - 1}}
	\end{equation}
\end{lemma}
\begin{proof}
	\textcolor{black}{	With $a_2$ and $b_2$ from (\ref{a2b2c2_values}), we get}
	\textcolor{black}{${r_2} = \sec \left( {\frac{{\beta \pi }}{2}} \right)$. Therefore,}
	\textcolor{black}{we have}	
	\begin{equation}
	\frac{{{c_2}}}{{{r_2}}} =  - \frac{{\sin \left( {\frac{{\beta \pi }}{2}} \right)}}{{{A_m} - 1}}
	\end{equation}
\end{proof}
\textcolor{black}{Note that $c_2/r_2$ in (\ref{c2_by_r2_final}), is a negative quantity because $\beta\in(0,2)$ and $A_m>1$.}
\begin{lemma}[Simplification of $\alpha_2$]
	$\alpha_2$ in (\ref{wp_equation_initial_simplified}) can be simplified as
	\begin{align} \label{alpha2_final}
	{\alpha _2} = \frac{\pi }{2} - \frac{{\beta \pi }}{2}
	\end{align}
	where ${\alpha _2} = {\tan ^{ - 1}}\left( {\frac{{{b_2}}}{{{a_2}}}} \right)$.
\end{lemma}
\begin{proof}
	We have $\alpha_2=\tan^{-1}(b_2/a_2)$. Substituting $a_2$ and $b_2$ from (\ref{a2b2c2_values}), we get
	\begin{equation} \label{alpha2_simplified_2}
	{\alpha _2} = {\tan ^{ - 1}}\left( {\frac{1}{{\tan \frac{{\beta \pi }}{2}}}} \right) = 
	{\tan ^{ - 1}}\left( {\cot \frac{{\beta \pi }}{2}} \right)
	\end{equation}
	Using \textcolor{black}{trigonometric identity}, $\cot(x) = \tan(\pi/2 - x) \forall x\in\Re$, we get
	\[{\alpha _2} = {\tan ^{ - 1}}\left( {\tan \left( {\frac{\pi }{2} - \frac{{\beta \pi }}{2}} \right)} \right)\]
	In (\ref{alpha2_simplified_2}), for $\beta\in(0,2)$, the argument lies in $(-\pi/2,\pi/2)$. Now, $\tan^{-1}(\tan( x )) = x$ when $x\in [-\pi/2, \pi/2]$. Therefore,
	\begin{align} \label{alpha_2_simplified_in_beta}
	{\alpha _2} = \frac{\pi }{2} - \frac{{\beta \pi }}{2}
	\end{align}
\end{proof}
\begin{theorem}[Solution of $\omega_{p}$]
	If  $\beta\in(0,2)$ and $A_m>1$, then
	\begin{equation} \label{wp_final_simplified}
	{\omega _p} = \frac{1}{\theta }\left( {\pi  - {{\cos }^{ - 1}}\left( {\frac{{\sin \left( {\frac{{\beta \pi }}{2}} \right)}}{{ {A_m}-1}}} \right) + \frac{\pi }{2} - \frac{{\beta \pi }}{2}} \right)
	\end{equation}
\end{theorem}
\begin{proof}
	The expression of $\omega_{p}$ can be found from (\ref{wp_equation_initial_simplified}) as
	\begin{equation} \label{wp_initial}
	{\omega _p} = \frac{1}{\theta }\left( {{{\cos }^{ - 1}}\left( {\frac{{{c_2}}}{r_2}} \right) + {\alpha _2}} \right)
	\end{equation}
	Let us assume ${\cos ^{ - 1}} (c_2/r_2) = \zeta$ and from (\ref{c2_by_r2_final}), $c_2/r_2 < 0$. Therefore, using inverse trigonometric property, $\cos^{-1}(-x) = \pi - \cos^{-1}(x); \, x>0$, we have
	\begin{equation} \label{zeta}
	\zeta  = \pi  - {\cos ^{ - 1}}\left( {\frac{{\sin \left( {\frac{{\beta \pi }}{2}} \right)}}{{{A_m} - 1}}} \right)
	\end{equation}
	From (\ref{wp_initial}), (\ref{zeta}) and (\ref{alpha2_final}) we have
	\begin{equation} \label{wp_final_in_zeta_alpha2}
	{\omega _p} = \frac{1}{\theta }\left( {\zeta  + \alpha_2 } \right)
	\end{equation}
	Substituting $\zeta$ from (\ref{zeta}) and $\alpha_2$ from (\ref{alpha2_final}), we get
	\begin{equation} 
	{\omega _p} = \frac{1}{\theta }\left( {\pi  - {{\cos }^{ - 1}}\left( {\frac{{\sin \left( {\frac{{\beta \pi }}{2}} \right)}}{{ {A_m}-1}}} \right) + \frac{\pi }{2} - \frac{{\beta \pi }}{2}} \right)
	\end{equation}
\end{proof}
\begin{remark} \label{remark_1}
	For $A_m = 1$, $c_2/r_2 \to \infty$ . Therefore, $A_m = 1$ cannot be chosen as controller design specification. Whereas, $A_m<1$ signifies an unstable closed loop system, thus irrelevant from control point of view.
\end{remark}

\subsection{Finding $\beta_{\omega_{p}}$ from set $\Xi_b$: 
}
For any practical control system, we have $\omega_p \in \Re^{+}$. 
Let ${\rm  B}_{{\omega _p}}$ is the set of those $\beta$ which satisfies some given $\omega_{p}$. Therefore, it is justified to write ${{\rm B} _{{\omega _p}}} = \left\{ {\beta :{\omega _p} \in {\Re ^ + }} \right\}$.

In (\ref{wp_final_simplified}), $\omega_p$ is function of $\theta$, $\beta$ and $A_m$, where $\theta$ is the process delay, $A_m$ is the desired gain margin and $\beta\in(0,2)$. 
The solution can be found in two steps. First, find ${{\rm B}_{{\omega _p}\Re }} = \left\{ {\beta :{\omega _p} \in \Re } \right\}$ and second, find ${{\rm B}_{{\omega _p} + }} = \left\{ {\beta :{\omega _p} > 0} \right\}$. Then, ${{\rm B}_{{\omega _p}}} = {{\rm B}_{{\omega _p}\Re }} \cap {{\rm B}_{{\omega _p} + }}$.
\begin{theorem}[Condition for $\omega_p\in\Re$]
	If $\omega_{p} \in \Re $ is as given in (\ref{wp_initial}), where $c_2/r_2$ is as given in (\ref{c2_by_r2_final}), then for $\beta\in(0,2)$ and $A_m \ge2 $, ${{\rm B} _{{\omega _p} + }} = \left\{ {\beta\in(0,2) :{\omega _p} > 0} \right\}$.
\end{theorem}
\begin{proof}
	For $\omega_p \in \Re$, $\left| c_2/r_2 \right| \le 1$. Hence, (\ref{c2_by_r2_final}) can be written as
	\begin{equation} \label{wp_in_real_initial}
	\left|\frac{\sin\left( \frac{\beta \pi}{2} \right)}{A_m-1} \right| \le 1
	\end{equation}
	For $\beta\in(0,2)$, $\sin\left(\beta\pi/2 \right) > 0$ and also for $A_m\ge 2$, $A_m-1>0$. Therefore, the modulus sign can be eliminated in in left hand side in  (\ref{wp_in_real_initial}), thus, we have	
	\begin{equation} \label{wp_in_real_final_in_mod_form}
		\sin \left( {\frac{{\beta \pi }}{2}} \right) \le  { {A_m} -  1}
	\end{equation}
	\begin{figure}
		\centering
		\includegraphics[width=2in] {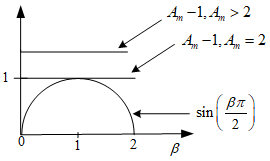}
		\caption{Graphical representation of (\ref{wp_in_real_final_in_mod_form})}
		\label{fig_wp_real_graphical}
	\end{figure}
	In Fig.\ref{fig_wp_real_graphical}, LHS and RHS of (\ref{wp_in_real_final_in_mod_form}) is plotted for $\beta \in (0,2)$ and $A_m\ge2$ respectively.
	From the figure, it is clear that the condition in (\ref{wp_in_real_final_in_mod_form}), i.e., for $A_m \ge 2$, $\omega_p \in \Re \, \forall \, \beta \in (0,2)$.
\end{proof}

\begin{theorem}[Condition for $\omega_p > 0$ ]
	If $\beta\in(0,2)$ and $A_m\ge1$, then  ${{\rm B} _{{\omega _p} + }} = \left\{ {\beta\in(0,2) :{\omega _p} > 0} \right\}$.
\end{theorem}
\begin{proof}
	Referring to (\ref{wp_final_in_zeta_alpha2}), for  $\omega_p > 0$, we need
	\begin{equation} \label{wp_more_zero_initial}
	\zeta>-\alpha_2
	\end{equation}
	where $\zeta$ is given in (\ref{zeta}) and $\alpha_2$ is given in (\ref{alpha_2_simplified_in_beta}). Substituting $\zeta$ and $\alpha_2$, \textcolor{black}{(\ref{wp_more_zero_initial}) can be written as}
	\begin{equation} \label{wp_more_zero_initial2}
	\pi  - {\cos ^{ - 1}}\left( {\frac{{\sin \left( {\frac{{\beta \pi }}{2}} \right)}}{{{A_m} - 1}}} \right) >  - \frac{\pi }{2} + \frac{{\beta \pi }}{2}
	\end{equation}
	In (\ref{wp_more_zero_initial2}), let us denote,
	$\gamma _1 = \pi  - {\cos ^{ - 1}}\left( {\frac{{\sin \left( {\frac{{\beta \pi }}{2}} \right)}}{{{A_m} - 1}}} \right)  $ 
	and 
	$ \gamma _2 = - \frac{\pi }{2} + \frac{{\beta \pi }}{2} $.
	Therefore, we need to prove that $\gamma_1>\gamma_2$ for $A_m \ge 2$ and $\beta \in (0,2)$.
	
	For $\beta\in(0,1]$, $\frac{{\sin \left( {\beta \pi /2} \right)}}{{{A_m} - 1}} \in \left( {0,\frac{1}{{{A_m} - 1}}} \right]$ and for $\beta(1,2)$, $\frac{{\sin \left( {\beta \pi /2} \right)}}{{{A_m} - 1}} \in \left( {0,\frac{1}{{{A_m} - 1}}} \right)$. 
	However, ${\frac{1}{{{A_m} - 1}}} \in (0,1]$ for $A_m \ge 2$. Therefore, $\left( {\frac{{\sin \left( {{{\beta \pi }}/{2}} \right)}}{{{A_m} - 1}}} \right) \in (0,1]$ for all $\beta \in (0,2)$ and $A_m \ge 2$.	
	Hence, $\cos^{-1}\left( {\frac{{\sin \left( {{{\beta \pi }}/{2}} \right)}}{{{A_m} - 1}}} \right) \in [0,\pi/2)$ for $\beta\in(0,2)$ and $A_m \ge 2$, which results $\gamma_1\in(\pi/2,\pi]$.
	
	Similarly, for $\beta \in(0,2)$, $\gamma_2\in(-\pi/2,\pi/2)$. Therefore,  $\gamma_1>\gamma_2$ for all $\beta\in(0,2)$ and $A_m \ge 2$. Hence, it concludes that $\omega_p>0 \, \forall \, \{\beta\in(0,2), \, A_m \ge 2\}$.
\end{proof}
\textcolor{black}{From Theorem 6 and 7, we have ${{\rm B}_{{\omega _p}\Re }} = \left\{ {\beta  \in \left( {0,2} \right):{\omega _p} \in \Re } \right\}$ and \\
	${{\rm B}_{{\omega _p} + }} = \left\{ {\beta  \in \left( {0,2} \right):{\omega _p} > 0} \right\}$ respectively. Since ${{\rm B}_{{\omega _p}}} = {{\rm B}_{{\omega _p}\Re }} \cap {{\rm B}_{{\omega _p} + }}$, therefore, ${{\rm B}_{{\omega _p}}} = \left\{ {\beta  \in \left( {0,2} \right):{\omega _p} \in {\Re ^ + }} \right\}$.} 

From Theorem 4 and 7, it is evident that ${\rm B}_{{\omega _g}}  \subset {\rm B}_{{\omega _p} }$. Hence, while finding the solution, there is no need to find $\omega_{p}$ for all the values of $\beta\in(0,2)$.
\begin{remark} \label{remark_2}
	From (\ref{PM_condition_eq1_tosolve}) and (\ref{PM_condition_eq2_tosolve}), as $\phi_{m}$ increases, $\lambda_a = \lambda_1= \lambda_2$ increases and from (\ref{GM_condition_eq3_tosolve}) and (\ref{GM_condition_eq4_tosolve}), as $A_m$ increases, $\lambda_b = \lambda_3= \lambda_4$ also increases. This relation would be helpful to select desired $A_m$ and $\phi_{m}$, \textcolor{black}{in the situation when} $\lambda_a$ and $\lambda_b$ plots do not intersect. Suppose, $\lambda_a>\lambda_b\, \forall \, {\rm B}_{{\omega _g} }$ and no intersection happens, then reducing $\phi_{m}$ and/or increasing $A_m$ may result in intersection of $\lambda_a$ and $\lambda_b$ and the solution can be obtained. 
\end{remark}
\begin{remark} \label{remark_3}
	\textcolor{black}{
		A similar relationship can be found between $\beta$ vs. $\phi_{m}$ and $\beta$ vs. $A_m$. Plotting the relationship as in (\ref{wg_from_PM_condn}), $\beta$ could be seen to decrease with increase in $\phi_{m}$ and plotting (\ref{wp_from_GM_condn}), $\beta$ could be seen to decrease with increase in $A_m$. Hence, $\beta$ is inversely related to $A_m$ as well as $\phi_{m}$.}
\end{remark}
\subsection{Disturbance rejection analysis}
It is easy to prove that the proposed FO controller control can reject disturbances. The theoretical analysis can be done by finding output $Y(s)$ in terms of input $R(s)$ and disturbance $D(s)$ \cite{MMorari_RobustPC_Book_1989}.
From Fig.\ref{figure_IMC_strucutre}, $Y(s) = \eta(s)R(s) + \epsilon(s))D(s)$ where $\eta(s) = C(s)G_p(s)/(1+C(s)G_p(s)) = L(s)/(1+L(s))$ and $\epsilon(s) = 1/(1+C(s)G_p(s)) = 1/(1+L(s))$. 
For disturbance rejection, it is sufficient to show $\mathop {\lim }\limits_{s \to 0} \epsilon(s) = 0$. Using (7), we have $C(s)G_p(s) = L(s) = \frac{{{e^{ - \theta s}}}}{{\lambda {s^\beta } + 1 - {e^{ - \theta s}}}}$. Since, $\mathop {\lim }\limits_{s \to 0} L(s) = \infty $, we get, $\mathop {\lim }\limits_{s \to 0} \epsilon(s) = 0$ which guarantees disturbance rejection.
\subsection{Algorithm for determining FO-IMC controller parameters} \label{Section_Cdesign_steps}
Here we present the procedure to find solution of $\beta$ and $\lambda$ as per the discussed theory in a step-by-step manner.
\begin{enumerate}
	\item [\textbf{Step-1:}] Select a desired $A_m$ and $\phi_{m}$.
	\item [\textbf{Step-2:}] Find range of $\beta_{\omega_{g}}$ such that $\omega_g\in\Re^+$ using Theorem 4. Define	a $(n  \times 1)$ array of $\beta_{\omega_{g}}$, where $'n'$ is \textcolor{black}{large enough} so that $\beta_{{\omega _g} }$ is almost continuous.
	\item [\textbf{Step-3:}] Find array $\omega_g$ corresponding to each value of $\beta_{\omega_{g}}$ using (\ref{wg_sol_final_expanded}). So $\omega_g$ is also a array of size $(n \times 1)$.
	\item [\textbf{Step-4:}] Consider $\beta_{\omega_{p}} = \beta_{{\omega _g}}$.
	\item [\textbf{Step-5:}] Find array $\omega_p$ corresponding to each value of $\beta_{\omega_{p}}$ using (\ref{wp_final_simplified}). So $\omega_p$ will also be an array of size $(n \times 1)$.
	\item [\textbf{Step-6:}] Find $\lambda_a$ either from (\ref{PM_condition_eq1_tosolve}) or (\ref{PM_condition_eq2_tosolve}) for each \textcolor{black}{$\beta_{{\omega _g} }(k)$} and $\omega_g(k)$. $\lambda_a$ has size $(n \times 1)$. \textcolor{black}{Here $k$ is the sample number.}
	\item [\textbf{Step-7:}] Find $\lambda_b$ either from (\ref{GM_condition_eq3_tosolve}) or (\ref{GM_condition_eq4_tosolve})  for each $\beta_{\omega_{p}}(k)$ and $\omega_p(k)$. $\lambda_b$ will be an array of $(n \times 1)$.
	\item [\textbf{Step-8:}] Plot for $\beta_{\omega_{g}}$ vs. $\lambda_a$ and $\beta_{\omega_{g}}$ vs. $\lambda_b$ in the same figure. The intersection point of two curves will satisfy the given $\phi_{m}$ and $A_m$ 
	specifications simultaneously. Let us say the intersection point is $\beta^{\star}$ and $\lambda^{\star}$ and occurs at sample point $k^{\star}$.
	\item [\textbf{Step-9:}] The system will have $\omega_g^{\star}=\omega_g(k^{\star})$ and $\omega_p^{\star}=\omega_p(k^{\star})$. This can easily be found by referring to the $k^{\star th}$ position of vector $\omega_g(k)$ and $\omega_p(k)$.
	\item [\textbf{Step-10:}] If intersection doesn't occur then change desired $\phi_m$ and $A_m$ as per guidelines given in \textcolor{black}{Remark \ref{remark_2}}.
\end{enumerate}

\section{Examples}
Two different FOPTD process models are considered to show the scope of the proposed controller.  First is a lag dominant model $(\tau>>\theta)$ taken from \cite{WKHo_IMC_PID_GPM_2001} and second is a delay dominant model ($\tau<<\theta$) taken from \cite{YCTian_CompofDom_and_VDelay_1998}. The results are compared with the two available methods in \cite{IKaya_IMC_PI_GPM_2004} and \cite{WKHo_IMC_PID_GPM_2001}  which satisfies $\phi_{m}$ and $A_m$ simultaneously. These are IMC-PID and IMC-PI control technique respectively.
These are to be implemented in classical control structure as shown in Fig. \ref{figure_IMC_strucutre}(b).
The proposed FO-IMC controller is designed without any approximation of the delay term, therefore, this \textcolor{black}{should} be implemented in IMC structure as in Fig.\ref{figure_IMC_strucutre}(a).

\textcolor{black}{The desired specifications of} $A_m=3 \, (=9.54dB)$ and \textcolor{black}{$\phi_{m}=65\deg(=1.1345 rad)$} are considered. The possible $\phi_m$ according to \cite{WKHo_IMC_PID_GPM_2001} would be \textcolor{black}{$74.31\deg(=1.2970rad)$} whereas according to \cite{IKaya_IMC_PI_GPM_2004} it will be \textcolor{black}{$60\deg(=1.0472rad)$} \textcolor{black}{for the selected $A_m$.}
\subsection{Example-I}
A lag dominant process model is taken from \cite{WKHo_IMC_PID_GPM_2001}, which is
\begin{equation}
{G_p}\left( s \right) = \frac{{0.43}}{{148s + 1}}{e^{ - 40s}}
\end{equation}
\begin{figure} [!ht]
	\centering
	\begin{tabular}{ c c c }
		\renewcommand{\thesubfigure}{($a_1$)}
		\subfigure[]{\includegraphics[width=1.55in]{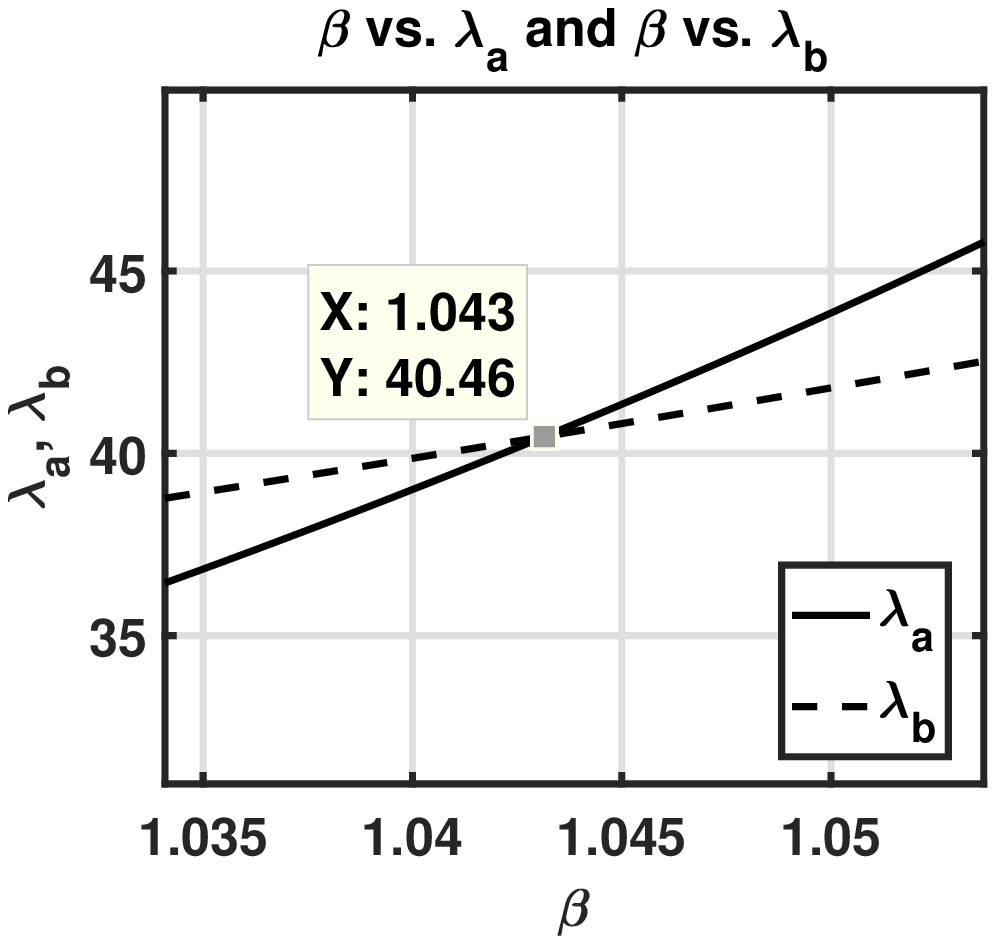}}
		&
				\renewcommand{\thesubfigure}{($b_1$)}
		\subfigure[]{\includegraphics[width=1.5in]{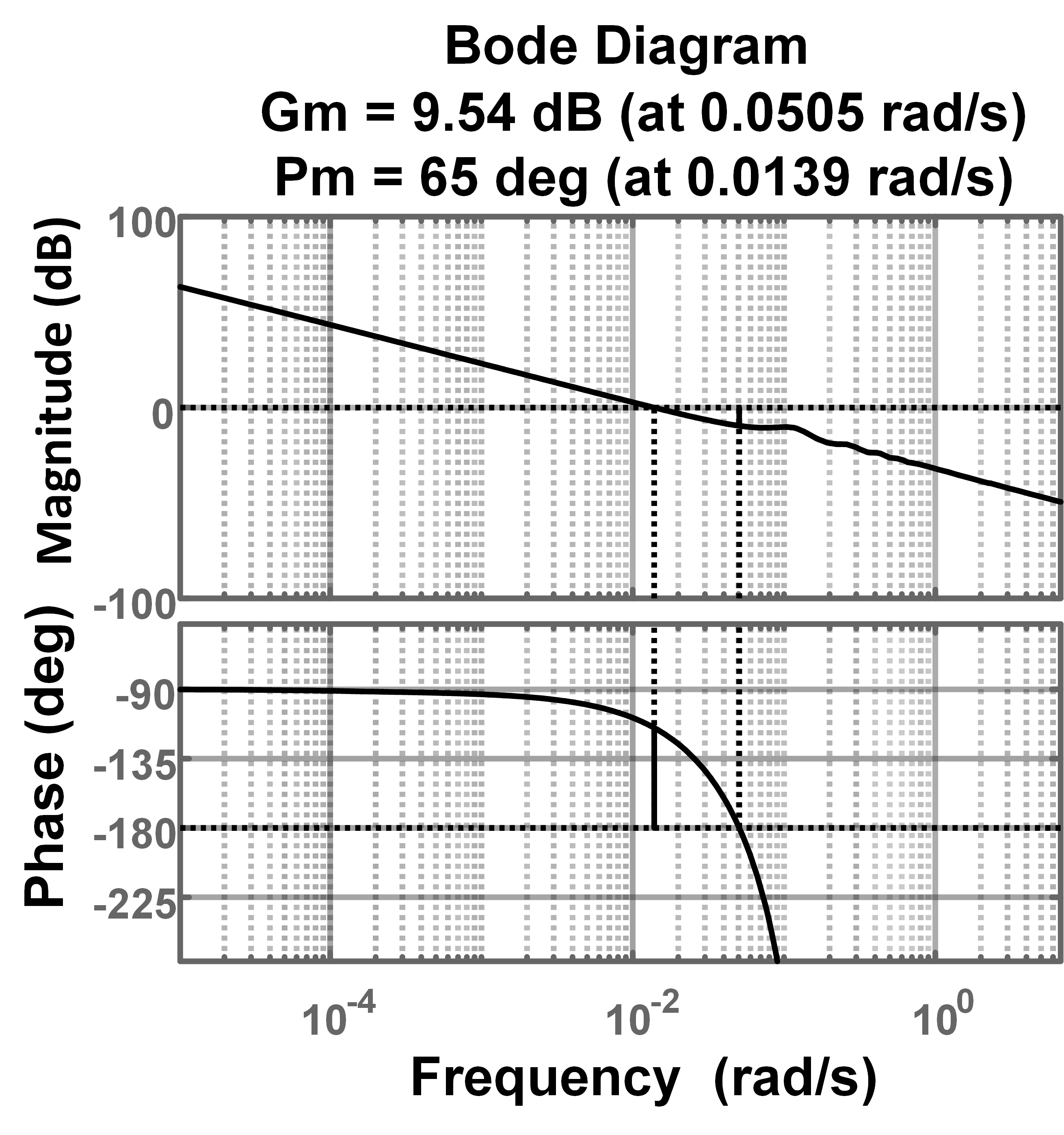}}
		&
				\renewcommand{\thesubfigure}{($c_1$)}
		\subfigure[]{\includegraphics[width=1.55in]{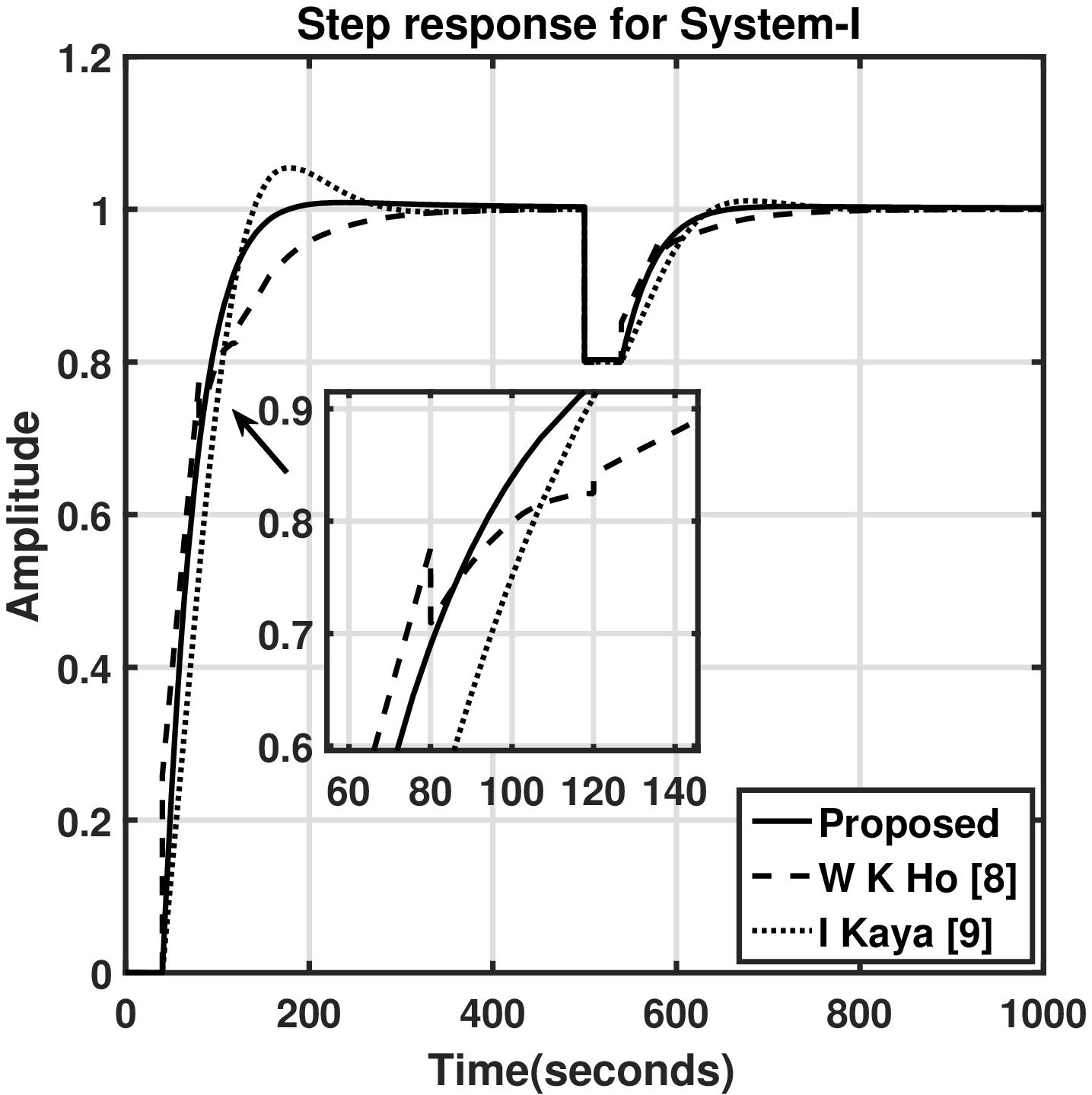}}
		\\
				\renewcommand{\thesubfigure}{($a_2$)}
		\subfigure[]{\includegraphics[width=1.5in]{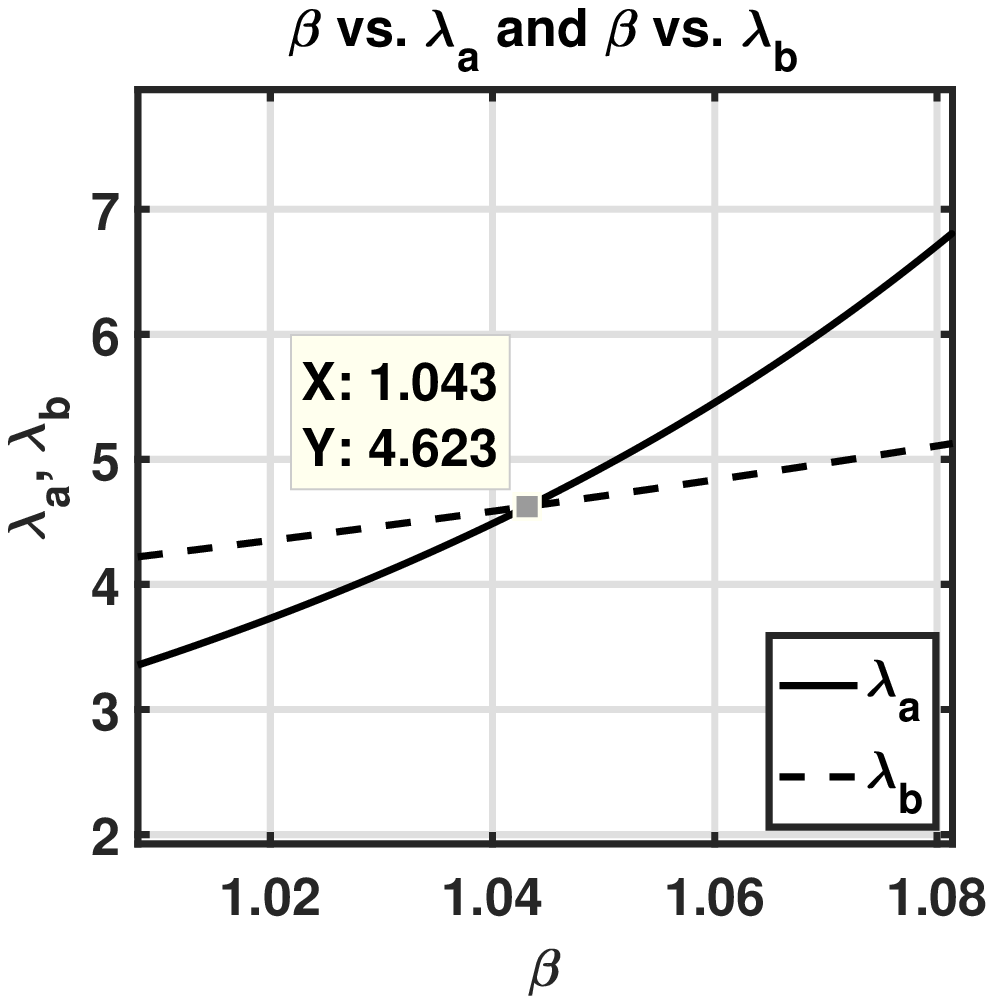}}
		&
		\renewcommand{\thesubfigure}{($b_2$)}
		\subfigure[]{\includegraphics[width=1.5in]{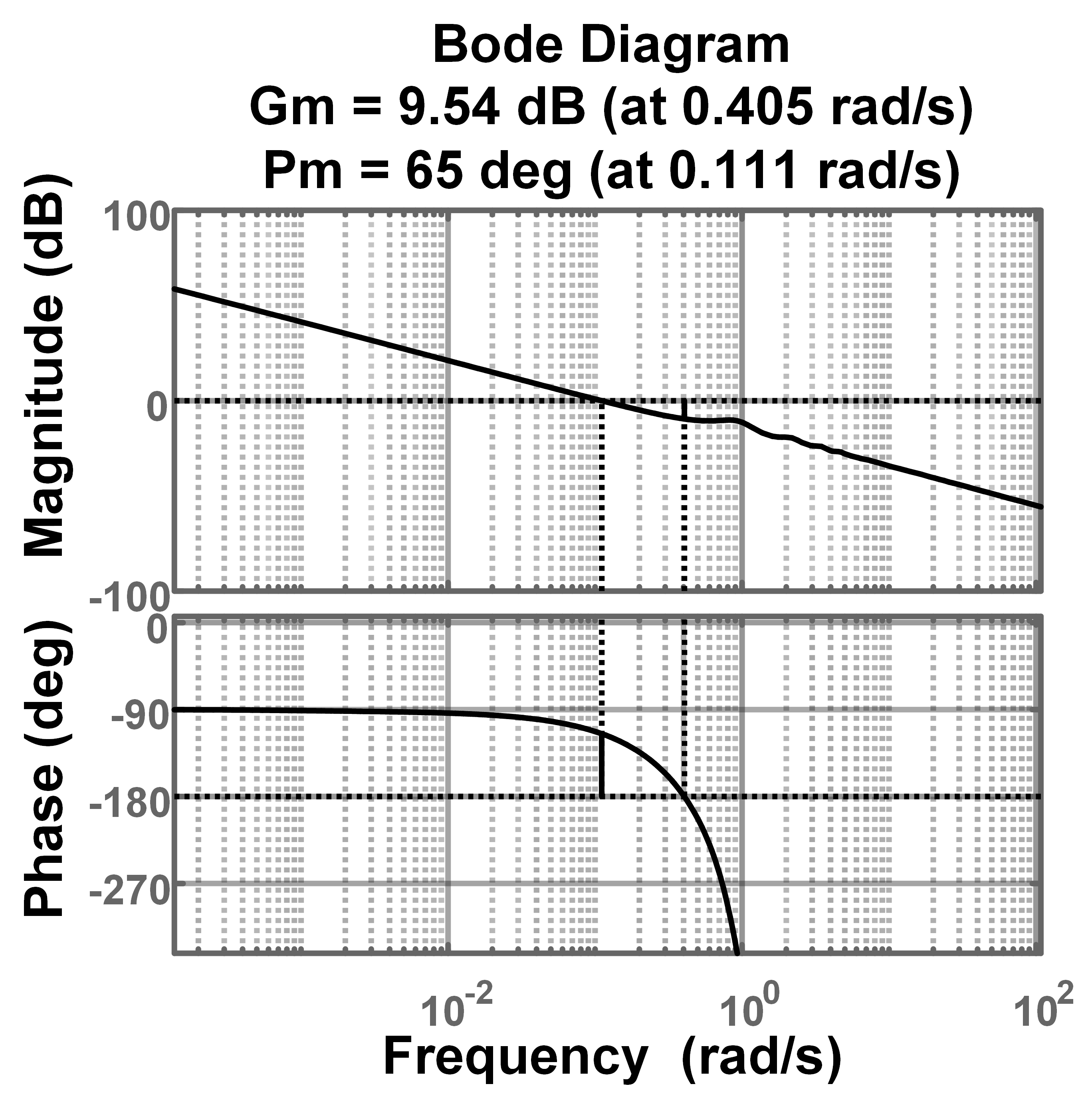}}
				&
		\renewcommand{\thesubfigure}{($c_2$)}
		\subfigure[]{\includegraphics[width=1.5in]{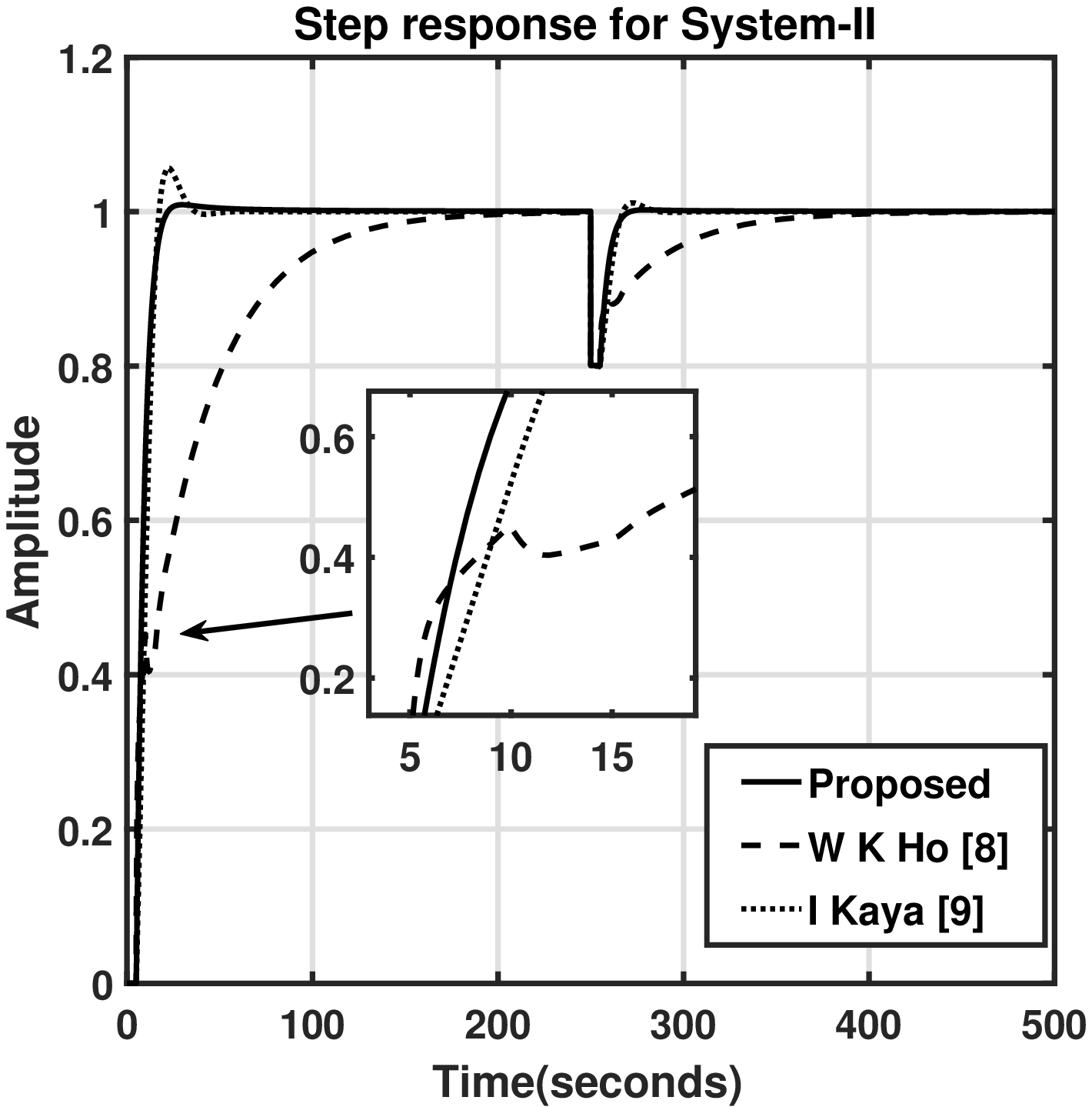}}		
	\end{tabular}
	\caption{Figure ($a_1$) and ($b_1$), $\beta\,vs.\,\lambda_a$ and $\beta\,vs.\,\lambda_b$. ($a_2$) and ($b_2$), Bode plot with proposed controller and ($a_3$) and ($b_3$), step response to compare different methods for system-I and II respectively. In $(b_1)$ and $b_2$, Gm is $A_m$ and Pm is $\phi_m$.}
	\label{fig_Exapmle_I_plots}
\end{figure}
Desired $A_m = 3$ and desired \textcolor{black}{$\phi_m = 65\deg (=1.1345 rad)$}.
Using controller design steps in Section-\ref{Section_Cdesign_steps}, $\beta_{x1} = 0.6389$ and $\beta_{y2}=1.2778$ are evaluated. Therefore, $\beta_{{\omega _g} }\in (0.6389,1.2778)$.
The controller parameters are obtained as $\beta^{\star} = 1.043$ and $\lambda^{\star} = 40.46$ (see Fig.\ref{fig_Exapmle_I_plots})($a_1$).
Corresponding to $\beta = 1.043$, we get $\omega_{g}^{\star} = 0.01391$ and $\omega_{p}^{\star} = 0.05066$.
The bode plot in Fig.\ref{fig_Exapmle_I_plots}($b_1$) shown that the simulation results are obtained as per theoretical calculations. 

For $A_m=3$, \cite{WKHo_IMC_PID_GPM_2001} will provide \textcolor{black}{$\phi_m = 74.31\deg (=1.2970rad)$} whereas \cite{IKaya_IMC_PI_GPM_2004} will provide \textcolor{black}{$\phi_m = 60 \deg(=1.0472rad)$}. In case of \cite{WKHo_IMC_PID_GPM_2001}, the obtained parameter are $\beta=1$ and $\lambda = 57.35$, which results in $PID = 5.05+0.03/s + 88.9945s$. Whereas, for \cite{IKaya_IMC_PI_GPM_2004}, $\beta=1$ and $\lambda=36.39$, resulting in $PI = 4.5056+148/s$. The step response is shown in Fig.\ref{fig_Exapmle_I_plots}($c_1$). \textcolor{black}{It is clear from the plot that the proposed method works best among the three.}
\subsection{Example-II}
A delay dominant process model is considered from \cite{YCTian_CompofDom_and_VDelay_1998}, which is
\begin{equation}
{G_p}\left( s \right) = \frac{{1}}{{0.5s + 1}}{e^{ - 5s}}
\end{equation}
The controller is designed for $A_m=3$ and $\phi_{m}=65\deg (=1.1345 rad)$.  Following steps in Section-\ref{Section_Cdesign_steps}, $\beta_{x1} = 0.6389$ and $\beta_{y2}=1.2778$ are evaluated. Therefore, $\beta_{{\omega _g} }\in(0.6389,1.2778)$. 
The controller parameters obtained are $\beta^{\star}=1.043$ and $\lambda^{\star} = 4.623$ (see Fig.\ref{fig_Exapmle_I_plots}($a_2$)). 
Corresponding to $\beta = 1.043$, we get $\omega_{g}^{\star} = 0.111$ and $\omega_{p}^{\star} = 0.405$.
The bode plot in Fig.\ref{fig_Exapmle_I_plots}($b_2$) shows that the simulation results are obtained as per the theoretical calculations. 

Corresponding to $A_m=3$, according to \cite{WKHo_IMC_PID_GPM_2001} the only possible $\phi_{m}$ is $ = 74.31\deg (=1.2970rad)$ and according to \cite{IKaya_IMC_PI_GPM_2004} it is $60\deg (=1.0472rad)$.
For \cite{WKHo_IMC_PID_GPM_2001}, $\beta=1$ and $\lambda = 7.1691$, giving $PID = 0.3103+0.1034/s+0.1293s$ and for \cite{IKaya_IMC_PI_GPM_2004}, $\beta=1$ and $\lambda = 4.5433$, giving $PI = 0.0524 + 0.5/s$. The step response is shown in Fig.\ref{fig_Exapmle_I_plots}($c_2$). \textcolor{black}{It is clear from the plot that the proposed method works best among the three.}

\section{Discussion and Conclusion}
An FO-IMC based controller is designed for desired gain margin and phase margin for a FOPTD process model. FO filter is used instead of IO filter \textcolor{black}{in the IMC structure} as it provides an additional parameter to tune as compared to \textcolor{black}{a single parameter in} IO filter. With only one tuning parameter in IO filter, the range of selection of desired $\phi_m$ and $A_m$ \textcolor{black}{is limited} to a 1-D curve. With two tuning parameters in FO filter, the range of selection of desired $\phi_{m}$ and $A_m$ becomes a 2-D surface. Therefore, they can be chosen independently.
The controller is designed without any approximation of the delay term appearing in the model of plant. 
To the best of authors' knowledge, this is the first attempt made in IMC literature where the controller is designed without any approximation of the delay term in the process model.

The proposed control strategy \textcolor{black}{should} be implemented in original IMC structure Fig.\ref{figure_IMC_strucutre}(a).
The controller is able to satisfy $\phi_m\in(0,\pi)$, $A_m>1$ for any $\theta/\tau>0$ which highly enhances the scope of the proposed control strategy.

\section*{Acknowledgment}
The authors are thankful to Prof. Ganti P. Rao, member of the UNESCO-EOLSS Joint Committee and Prof. Shankar P. Bhattacharya, TEXAS A\&MU, USA, for useful discussion at the initial stages of this work.

%
%
%
%

\bibliographystyle{IEEEtran}
\bibliography{mybibfile}
\end{document}